\documentclass[twocolumn]{autart}    

\usepackage{graphicx}          
\usepackage{amsmath}
\usepackage{amssymb}
\usepackage{mathrsfs}
\usepackage{mathtools}
\usepackage{bm}
\usepackage{amsfonts}
\usepackage{cases}
\usepackage{cite}
\usepackage{multirow}
\usepackage{color}
\usepackage{algorithm}
\usepackage{algpseudocode}
\usepackage{subcaption}
\usepackage{epstopdf}
\usepackage{lipsum}
\usepackage{url}
\usepackage{tikz}
\usetikzlibrary{positioning,shapes}
\graphicspath{{./}{Fig/}}

\let\theoremstyle\relax  
\usepackage{amsthm}  

\newtheorem{definition}{Definition}
\newtheorem{lemma}{Lemma}

\newtheorem{theorem}{Theorem}
\newtheorem{remark}{Remark}

\newtheorem{assumption}{Assumption}

\theoremstyle{definition}
\newtheorem*{example*}{Example}

\begin{document}

\begin{frontmatter}

\title{Collaborative design of fault diagnosis and fault tolerance control under nested signal temporal logic specifications} 

\author[Paestum]{Penghong Lu},
\author[Paestum]{Gang Chen},
\author[Baiae]{Rong Su}
\address[Paestum]{Shien-Ming Wu School of Intelligent Engineering, South China University Of Technology, Guangzhou, China}  
 \address[Baiae]{School of Electrical and Electronic Engineering, Nanyang Technological University, Singapore}        


\begin{abstract}                          
Signal Temporal Logic (STL) specifications play a crucial role in defining complex temporal properties and behaviors in safety-critical cyber-physical systems (CPS). However, fault diagnosis (FD) and fault-tolerant control (FTC) for CPS with nonlinear dynamics remain significant challenges, particularly when dealing with nested signal temporal logic (NSTL) specifications. This paper introduces a novel framework for the collaborative design of FD and FTC, aimed at optimizing fault diagnostic performance while ensuring fault tolerance under NSTL specifications. The proposed framework consists of four key steps: (1) construction of the Signal Temporal Logic Tree (STLT), (2) fault detection via the construction of fault-tolerant feasible sets, (3) evaluation of fault detection performance, and (4) synthesis of fault-tolerant control. Initially, a controller for nonlinear systems is designed to satisfy NSTL specifications, and a fault detection observer is developed alongside fault-tolerant feasible sets. To address the challenge of maintaining solution feasibility in dynamic optimization control problems, the concept of fault-tolerant control recursive feasibility is introduced. Subsequently, suboptimal controller gains are derived through a quadratic programming approach to ensure fault tolerance. The collaborative design framework enables more rapid and accurate fault detection while preserving FTC performance. A simulation study is presented to demonstrate the effectiveness of the proposed framework.
\end{abstract}

\end{frontmatter}

\section{Introduction}
Cyber-Physical Systems (CPS) represent a class of complex engineering systems that tightly integrate computation and physical processes to achieve high precision, adaptability, and efficiency \cite{1}. These systems are increasingly being deployed in critical infrastructures, such as intelligent transportation networks, automated production lines, and medical devices, where safety is paramount. To maintain safety, it is essential to continuously monitor the system's operational status and implement timely corrective measures upon detecting deviations from expected behavior \cite{4,5}. They are required to operate safely and reliably, even in the presence of faults or disturbances, highlighting the importance of robust fault diagnosis and fault-tolerant control mechanisms \cite{2,3}.

A key challenge in CPS is the formal specification of system requirements in a structured and verifiable manner. Temporal logics, such as Linear Temporal Logic (LTL) and Signal Temporal Logic (STL), have become essential tools for defining high-level system behaviors over time \cite{6,7,8}. These logics are particularly valuable in autonomous systems due to their systematic approach to specifying complex task requirements. LTL and STL allow system designers to capture desired properties of CPS through formal descriptions of event sequences and signal values.

In particular, STL provides a robust framework for analyzing real-time signals, making it especially useful in the monitoring and control of safety-critical systems. Since its introduction by Maler et al. \cite{9}, STL has seen widespread application in these areas, enabling precise specification of time-sensitive behaviors. The use of temporal logic in CPS is crucial for ensuring system reliability and safety. By allowing real-time monitoring and fault detection based on formal specifications, temporal logics ensure that system behaviors adhere to safety and performance requirements under both normal and faulty conditions. This makes them indispensable for CPS, where failure to promptly detect and correct faults can lead to catastrophic consequences.

In recent years, several online monitoring algorithms have been developed for real-time verification of system behavior against temporal logic specifications \cite{10,11,12}. These algorithms are critical for continuously assessing whether systems adhere to safety requirements. Leucker \cite{13} introduced predictive semantics for monitoring non-black-box systems with non-timed LTL specifications, while Pinisetty et al. \cite{14} proposed a predictive runtime verification framework for systems with timing constraints. More recently, Ferrando et al. \cite{15} extended predictive monitoring to centralized and multi-model scenarios. Yoon et al. \cite{16} introduced a Bayesian intent inference framework that uses robot intent information to predict future positions. For data-driven STL-based predictive monitoring, Qin et al. \cite{17} employed statistical time series analysis to predict future states, and Ma et al. \cite{18} used Bayesian recurrent neural networks to forecast uncertain future states. Lindemann et al. \cite{19} applied conformal prediction to provide probabilistic guarantees,  while Yu et al. \cite{8} used system models to predict future states when evaluating STL formula satisfaction. 

Additionally, there is a significant increase in interest in control techniques for CPS that provide formal guarantees of safety and liveness \cite{20}, where temporal logic is used as the specification language. Rigorous methods for designing, verifying, and implementing CPS controllers have been proposed to ensure that safety and liveness are always met, even in changing or unpredictable environments. Many of the developed methods are based on safety/liveliness filters or automata-based temporal logic approaches \cite{21}. For CPS, there are numerous methods based on safety/liveliness filters, such as methods based on Hamilton-Jacobi reachability analysis \cite{22}, control barrier functions \cite{23,24}, and zonotope-based methods \cite{25}. In literature \cite{26}, the authors explored the use of Hamilton-Jacobi reachability analysis to synthesize control sets that satisfy STL specifications. In literature \cite{27}, the authors explored the application of control barrier functions to efficiently synthesize control strategies for STL fragments. Recently, literature \cite{28,29} introduced a tree-based computational model to address nested temporal operators in the temporal logic specifications, called Temporal Logic Trees, which directly utilizes backward reachability analysis to verify and synthesize control sets for LTL and STL.

While significant progress has been made in online monitoring algorithms and control synthesis techniques for cyber-physical systems (CPS), a critical gap remains in their integration, particularly for fault-tolerant control (FTC) \cite{30,31,32,33} involving temporal logic specifications with nested operators. For example, \cite{10} proposed a model predictive monitoring method based on Signal Temporal Logic (STL), optimizing system behavior through model prediction to satisfy temporal logic specifications. Despite its real-time performance and fault prediction capabilities, this method is limited to non-nested STL specifications, making it inadequate for capturing complex temporal dependencies and dynamic system behaviors with multi-level constraints or nested events. Additionally, the approach suffers from insufficient coupling between fault detection and control strategies, hindering prompt adjustments during faults. In contrast, \cite{8} introduced a continuous-time control synthesis method for nested temporal logic specifications, leveraging Signal Temporal Logic Trees (STLT) to handle complex temporal constraints. While STLT effectively represents nested structures, it struggles with dynamic optimization feasibility, particularly in real-time applications, and lacks robustness in addressing the complexity of nonlinear systems. Consequently, neither approach resolves the FTC problem for nonlinear systems with nested temporal logic specifications, which remains an open and critical challenge in the field. 

The separation of fault diagnosis from control synthesis is a significant limitation, as accurate fault detection is crucial for enabling FTC mechanisms to maintain system stability and performance in the presence of faults. Therefore, integrating fault monitoring with control synthesis through temporal logic frameworks, especially STL, is essential for building CPS that are not only fault-aware but also capable of adapting to ensure continuous safe operation. Addressing this research gap is vital for advancing the resilience and reliability of CPS in safety-critical applications.

In this paper, we present a novel Collaborative Design (CoD) methodology that integrates condition monitoring and control synthesis into a unified framework, using STL as the specification language. This unified approach enables real-time, precise evaluation of system behavior, seamlessly integrating fault detection and control synthesis into a cohesive structure. This approach capitalizes on the strengths of both fault monitoring and control synthesis to enhance system reliability and performance. The key contributions of this work are as follows:
\begin{itemize}
    \item We  extend existing STL-based monitoring techniques to handle systems with nested temporal logic specifications with the help of a practical tool called Signal Temporal Logic Tree (STLT) \cite{8}. This provides a more comprehensive solution for fault management in CPS.
    \item We propose an online monitoring method that incomathcal Xorates system dynamics, enabling more accurate real-time assessments of task completion. This enhances the precision of condition monitoring and fault detection.
    \item  We synthesize   FTC strategies with designed adaptive control barrier functions (CBFs) and their activation time intervals within the STLT framework. This allows for the online synthesis of control strategies that ensure both safety and task completion, while dynamically adapting to system faults.
    \item  To address the challenge of ensuring a feasible solution in dynamic optimization control problems, the concept of fault-tolerant control recursive feasibility is introduced. This concept pertains to the control strategy's ability to maintain the system's controllability in the presence of faults, ensuring global stability and constraint satisfaction under fault conditions. 
\end{itemize}
 
The rest of the paper is structured as follows: Section II presents the problem statement and preliminaries. Section III introduces the proposed CoD algorithm and STLT framework. Section IV provides a simulation example, and Section V concludes with remarks on future research directions.

\section{Problem statement and preliminary}

\subsection{Systems dynamics}

\textbf{Notation.} Define $\mathbb{R}:=(-\infty, \infty)$, $\mathbb{R}_{\ge 0}:=[0, \infty)$. Let $\mathbb{N}$ be the set of natural numbers. Denote $\mathbb{R}^n$ as the $n$ dimensional real vector space, $\mathbb{R}^{n\times m}$ as the $n\times m$ real matrix space. Define $\|x\|$ and $\|A\|$ as the Euclidean norm of vector $x$ and matrix $A$. Given a set $S\subset \mathbb{R}^n$, $\overline{S}$ denotes its complement. Given a vector $x\in \mathbb{R}^n$ and a set $S\subset \mathbb{R}^n$, the distance function is defined as $\texttt{dist}(x, S):=\inf_{y\in S}\|x-y\|$. 

Consider a discrete-time dynamical system of the form
\begin{equation}\label{x0}
x_{k+1} =f(x_k, u_k),\\
\end{equation}
where $x_k:=x(t_k)\in \mathbb{R}^n, u_k:=u(t_k)\in U, k\in \mathbb{N}$ are the state and control input of the system, respectively. The time sequence $\{t_k\}$ is a sequence of sampling instants, which satisfy $t_k\geq 0$. The control input is constrained to a compact set $U\subset \mathbb{R}^m$. The function $f: \mathbb{R}^n\times \mathbb{R}^m\to \mathbb{R}^n$ is the dynamic function of the system in $x$ and $u$. Denote by $\mathcal{U}_{\ge k}$ the set of all control policies that start from time $t_k$.

\subsection{Signal temporal logic}

Previous studies use  STL  formulae with bounded-time temporal operators to describe whether or not the trajectory of the system satisfies some desired high-level properties. STL is a predicate logic consisting of predicates $\mu$, which are defined through a predicate function $g_\mu: \mathbb{R}^n\to \mathbb{R}$ as
\begin{equation*}
	\mu:=\left\{\begin{aligned}
		\top, & \quad \text{if } \quad g_\mu(\bm{x})\ge 0 \\
		\bot, & \quad \text{if } \quad g_\mu(\bm{x})<0.
	\end{aligned}\right.
\end{equation*}

Desired robot behaviors are specified using the following expressive STL syntax:
\begin{equation}\label{Def:PNF}
  \begin{aligned}
&\hspace{0cm}\varphi ::= \top \mid \mu \mid \neg \mu \mid \varphi_1 \wedge \varphi_2 \mid \varphi_1 \vee \varphi_2 \mid  \varphi_1\mathsf{U_{II}} \varphi_2\mid  \mathsf{F_{II}} \varphi\mid  \mathsf{G_{II}} \varphi,
 \end{aligned}
\end{equation}
where $\varphi, \varphi_1, \varphi_2$ are STL formulas, and $\mathsf{II}=[a,b]$ is a time interval, b$>$a $\geq$ 0. Here, $\wedge$ and $\vee$ are logic operators ``conjunction" and ``disjunction", $\mathsf{U}$, $\mathsf{F}$, and $\mathsf{G}$ are temporal operators ``until", ``eventually", and ``always", respectively.

The validity of an STL formula $\varphi$ with respect to a discrete-time
signal $\bm{x}$ at time $t_k$, is defined inductively as follows:
\begin{eqnarray*}
  (\bm{x}, t_k) \vDash \mu &\Leftrightarrow& g_\mu(\bm{x}(t_k))\ge 0, \\
   (\bm{x}, t_k) \vDash  \neg  \mu &\Leftrightarrow& \neg((\bm{x}, t_k) \vDash \mu), \\
   (\bm{x}, t_k) \vDash \varphi_1 \wedge \varphi_2 &\Leftrightarrow& (\bm{x}, t_k) \vDash \varphi_1 \wedge  (\bm{x}, t_k) \vDash \varphi_2, \\
   (\bm{x}, t_k) \vDash \varphi_1 \vee \varphi_2 &\Leftrightarrow& (\bm{x}, t) \vDash \varphi_1 \vee  (\bm{x}, t_k) \vDash \varphi_2, \\
   (\bm{x}, t_k) \vDash \varphi_1 \mathsf{U}_{[a, b]} \varphi_2 &\Leftrightarrow& \exists {t_{k'}}\in [t_k+a, t_k+b]  \ \text{s.t.} \\
   & &(\bm{x}, {t_{k'}}) \vDash \varphi_2 \wedge \\
   & &\forall {t_{k''}}\in [t_k, t_{k'}],(\bm{x}, {t_{k''}}) \vDash \varphi_1, \\
   (\bm{x}, t_k) \vDash \mathsf{F}_{[a, b]} \varphi &\Leftrightarrow& \exists {t_{k'}}\in [t_k+a, t_k+b]  \ \text{s.t.} \\
   & &(\bm{x}, {t_{k'}}) \vDash \varphi, \\
   (\bm{x}, t_k) \vDash \mathsf{G}_{[a, b]} \varphi &\Leftrightarrow& \forall {t'}\in [t_k+a, t_k+b]  \ \text{s.t.} \\
   & &(\bm{x}, {t_{k'}}) \vDash \varphi.
  \end{eqnarray*}

\begin{definition}(Satisfiability)\label{Def:satisfibility}
	Consider the deterministic system (\ref{x0}) and the STL formula $\varphi$ (\ref{Def:PNF}). We say $\varphi$ is satisfiable from the initial state $x_0$ if there exists a control signal $\bm{u}\in \mathcal{U}$ such that
	$({\bm{x}}_{x_0}^{\bm{u}}, 0) \vDash \varphi.$
\end{definition}

Given an STL formula $\varphi$, let
\begin{equation}\label{initialsatisfiableset}
	\mathbb{S}_{\varphi}:=\{x_0\in \mathbb{R}^n| \text{$\varphi$ is (robustly) satisfiable from $x_0$}\}
\end{equation}
denote the set of initial states from which $\varphi$ is (robustly) satisfiable.

\subsection{Reachability operators}
In this section, the study elaborates on two reachability operators. The inherent connection between reachability and temporal operators plays a crucial role in the approach proposed in this paper. The definitions for the maximum and minimum reachable sets are presented below \cite{chen2018signal}.

\begin{definition}\label{Def:maxreachset}
   Consider the system (\ref{x0}), a set $\mathcal{S} \subseteq \mathbb{R}^n$, and a time interval $[a, b]$. The \emph{maximal reachable set} $\mathcal{R}^M(\mathcal{S}, [a, b])$ is defined as
    \begin{multline*}
    \mathcal{R}^M(\mathcal{S},[a, b])=\Big\{x_0\in \mathbb{R}^n: \exists \bm{u}\in \mathcal{U}, \exists t_{k'}\in [a, b], \\
      \text{s.t. }\;  {\bm{x}}_{x_0}^{\bm{u}}(t_{k'})\in \mathcal{S} \Big\}.
    \end{multline*}
\end{definition}

The set $\mathcal{R}^M(\mathcal{S},[a, b])$ collects all states in $\mathbb{R}^n$ from which there exists a control policy $\bm{u}\in \mathcal{U}_{\ge k}$ that drives the system to target set $\mathcal{S}$ at some time instant $t_{k'}\in [a, b]$.

\begin{definition}\label{Def:minreachset}
   Consider the system (\ref{x0}), the set $\mathcal{S}\subseteq \mathbb{R}^n$, and a time interval $[a, b]$. The \emph{minimal reachable set} $\mathcal{R}^m(\mathcal{S},[a, b])$ is defined as
    \begin{multline*}
    \mathcal{R}^m(\mathcal{S},[a, b])=\Big\{x_0\in \mathbb{R}^n: \forall \bm{u}\in \mathcal{U}, \exists t_{k'}\in [a, b], \\
      \text{s.t. }\; {\bm{x}}_{x_0}^{\bm{u}}(t_{k'})\in \mathcal{S} \Big\}.
    \end{multline*}
\end{definition}

 The set $\mathcal{R}^m(\mathcal{S},[a, b])$ collects all states in $\mathbb{R}^n$ from which no matter what input signal $\bm{u}\in \mathcal{U}_{\ge k}$ is applied, the system can reach the target set $\mathcal{S}$ at some time instant $t_{k'}\in [a, b]$.
Define $\mathcal{S}$ as the zero superlevel set of a function: $\mathcal{S}=\{x\in \mathbb{R}^n: h_{\mathcal{S}}(x)\ge 0\}$. Similarly, define $\mathcal{R}^M(\mathcal{S},[a, b])$ and $\mathcal{R}^m(\mathcal{S},[a, b])$ as superlevel sets of some functions, i.e.,  	\begin{equation*}
\begin{aligned}
\mathcal{R}^M(\mathcal{S}, [a, b])&:=\{x: h_{\mathcal{R}^M(\mathcal{S}, [a, b])}(\bm{x})\ge 0\},\\
\mathcal{R}^m(\mathcal{S}, [a, b])&:=\{x: h_{\mathcal{R}^m(\mathcal{S}, [a, b])}(\bm{x})\ge 0\}.		\end{aligned}
\end{equation*} 
The calculation of the maximal and minimal reachable sets can be formulated as an optimal control problem \cite{chen2018signal}, as presented below:
	\begin{equation*}
		\begin{aligned}
			h_{\mathcal{R}^M(\mathcal{S}, [a, b])}&(x)=\max_{\bm{u}\in \mathcal{U}}\max_{s\in [a, b]}h_{\mathcal{S}}(\bm{x}_{x}^{\bm{u}}(s)),\\
			h_{\mathcal{R}^m(\mathcal{S}, [a, b])}&(x)=\min_{\bm{u}\in \mathcal{U}}\max_{s\in [a, b]}h_{\mathcal{S}}(\bm{x}_{x}^{\bm{u}}(s)).
		\end{aligned}
	\end{equation*}
In the following, relations are established between the STL temporal operators and the maximal/minimal reachable sets. 

\begin{lemma}
Consider the system (\ref{x0}) and predicates $ \mu_1,$ and $ \mu_2$. Then, one has \cite{chen2018signal}
	\begin{itemize}
		\item $\mathbb{S}_{\mu_1 \mathsf{U}_{[a, b]} \mu_2}(t_k)= \mathcal{R}^M(\mathbb{R}^n, \mathbb{S}_{\mu_2},\mathbb{S}_{\mu_1}, [a, b])$;
		\item $\mathbb{S}_{\mathsf{F}_{[a, b]} \mu_1}(t_k)= \mathcal{R}^M(\mathbb{R}^n, \mathbb{S}_{\mu_1}, [a, b])$;
		\item $\mathbb{S}_{\mathsf{G}_{[a, b]} \mu_1}(t_k)=  \overline{\mathcal{R}^m(\mathbb{R}^n, \overline{\mathbb{S}_{\mu_1}}, [a, b])}$,
	\end{itemize}
	where $\mathbb{S}_{\mu_1 \mathsf{U}_{[a, b]} \mu_2}, \mathbb{S}_{\mathsf{F}_{[a, b]}\mu_1}$ and $\mathbb{S}_{\mathsf{G}_{[a, b]}\mu_1}$ are the satisfying sets for $\mu_1 \mathsf{U}_{[a, b]} \mu_2$, $\mathsf{F}_{[a, b]}\mu_1$ and $\mathsf{G}_{[a, b]}\mu_1$, respectively.
\end{lemma}

\subsection{Nested STL formula and STLT}
 Nested STL is an extension of Signal Temporal Logic (STL), which enhances expressiveness, interpretability, and modularity by enabling the specification of hierarchical and complex temporal dependencies in a compact and intuitive manner. Nested STL extends STL by allowing the predicates themselves to be STL formulas with temporal operators, instead of simple atomic expressions.  Examples of nested STL formulas include $\mathsf{F}_{[a, b]}\varphi_1, \label{F}$,$\mathsf{G}_{[a, b]}\varphi_1,\label{G}$ and $\varphi_1\mathsf{U}_{[a, b]}\varphi_2$,  where $\varphi_1$ and at least one of $\varphi_1$ and $\varphi_2$ in   include temporal operators. In addition, $\varphi_1$  and  $\varphi_1, \varphi_2$  are called the \emph{argument(s)}  of the STL formula $\varphi$.

In order to encode the task satisfaction constraint (i.e., $(\bm{x}_{x_0}^{\bm{u_k}}, 0) \models \varphi$) as a set of constraints on the system input $u_k$ when $\varphi$ is nested, the STL tree encoding method proposed in \cite{8} is utilized, where an STLT refers to a tree with linked set nodes and operator nodes. The formal definition is given as follows.

\begin{definition}[STLT \cite{8}]
A \emph{STLT} is a tree for which the next holds:
\begin{itemize} 
\item Each node is either a \emph{set} node, which is a subset of $\mathbb{R}^n$, or an \emph{operator} node, which is an element of $\{\wedge, \vee, \mathsf{U}_{[a, b]},\mathsf{F}_{[a, b]},\mathsf{G}_{[a, b]}\}$;
\item the root node and the leaf nodes are \emph{set} nodes;
\item if a \emph{set} node is not a leaf node, its unique child is an \emph{operator} node;
\item the children of any \emph{operator} node are \emph{set} nodes.
  \end{itemize}
\end{definition}

Then the satisfaction relation between a given trajectory and a complete path of an STLT can be defined as follows.

\begin{definition}\label{Def:PathSaf}
	(Complete path of a STLT \cite{8}) Consider a trajectory $\bm{x}$ and encode a complete path $\bm{p}=\mathbb{X}_0\Theta_1\mathbb{X}_1\Theta_2\ldots \Theta_{N_f} \mathbb{X}_{N_f}$. It can say \emph{$\bm{x}$ satisfies $\bm{p}$ from time $t_k$}, denoted by $(\bm{x}, t_k)\cong \bm{p}$, if there exists a time coding for $\bm{p}$ such that $\underline{t}_{k}=\bar{t}_k=t_k$ and, for $i=1,2, \ldots, N_f $,
	\begin{itemize}
		\item[i)] if $\Theta_i\in \{\wedge, \vee\}$, then $[\underline{t}_{k_i}, \bar{t}_{k_i}]=[\underline{t}_{k_{i-1}}, \bar{t}_{k_{i-1}}]$;
		\item[ii)] if $\Theta_i\in \{\mathsf{U}_{[a, b]}, \mathsf{F}_{[a, b]}\}$, then $\exists t_{k'}\in [a, b]$ s.t. $[\underline{t}_{k_{i}}, \bar{t}_{k_{i}}]={t_{k'}}+[\underline{t}_{k_{i-1}}, \bar{t}_{k_{i-1}}]$;
		\item[iii)] if $\Theta_i=\mathsf{G}_{[a, b]}$, then $[\underline{t}_{k_{i}}, \bar{t}_{k_{i}}]=[a, b]+[\underline{t}_{k_{i-1}}, \bar{t}_{k_{i-1}}]$;and, for $i=0,1, \ldots, N_f $,
		\item[iv)] $\bm{x}(t_k)\in \mathbb{X}_{i}, \forall t_k\in [\underline{t}_{k_{i}}, \bar{t}_{k_{i}]}$.
	\end{itemize}
 where $[\underline{t}_{k_{i}}, \bar{t}_{k_{i}]}$ denote a time interval, $0 \leq \underline{t}_{k_{i}} < \bar{t}_{k_{i}}$ to each set node $\mathbb{X}_{i}$ in the complete path.
\end{definition}

Utilizing the above definition, the satisfaction relation between a trajectory $\bm{x}$ and an STLT can be expressed as follows.

\begin{definition}
    Consider a trajectory $\bm{x}$ and construct an STLT $\mathcal{T}_\varphi$. It can say \emph{$\bm{x}$ satisfies $\mathcal{T}_\varphi$}, denoted by $\bm{x} \vDash \mathcal{T}_\varphi$, if $\bm{x}$ satisfies at one or more complete paths of $\mathcal{T}_\varphi$. {\color{red} }
\end{definition}

The construction of an STLT  based on an STL  formula $\varphi$, as detailed in \cite{8}, involves a systematic procedure rooted in the syntax tree of $\varphi$. Using Algorithm 1 from \cite{8}, the STLT $\mathcal{T}_\varphi$ is generated, where the syntax tree nodes are categorized into \textbf{predicate nodes} (leaf nodes) and \textbf{operator nodes}(non-leaf nodes). Algorithm 2 in \cite{8} evaluates the truth value of the formula over a given trajectory $\bm{x}$. It takes three inputs: a trajectory $\bm{x}$, the STLT $\mathcal{T}_\varphi$, and a time-coding ${t_{\kappa_i}}$, and produces a Boolean result (`true' or `false').  This structured methodology efficiently evaluates STL formulas over given trajectories, leveraging the hierarchical representation of STLTs for interpretability and computational clarity.  The following example illustrates how to construct STLT for an STL formula.

\begin{example*} \label{ex:complete paths}
Given formula  $\varphi=\mathsf{G}_{[0, 16]}\mathsf{F}_{[2, 10]}\mu_1 \vee \mathsf{F}_{[0, 14]}( \mu_2  \mathsf{U}_{[5, 10]}\mu_3)$,  where $\mu_i, i=\{1,2,3\}$ are predicates, it can rewrite $\varphi$ into the desired form $\varphi=\mathsf{G}_{[0, 16]}\mathsf{F}_{[2, 10]}\mu_1 \vee \mathsf{F}_{[0, 14]} \mathsf{G}_{[0, 10]} \mu_2 \wedge \mathsf{F}_{[0, 14]}\mathsf{F}_{[5, 10]} \mu_3$. The constructed STLT $\mathcal{T}_{\hat\varphi}$ is plotted in Fig. \ref{Fig:1}, which follows a bottom-up approach. We begin by constructing the leaf nodes corresponding to the three predicates: $\mathbb{X}_5=\mathbb{S}_{\mu_1}$, $\mathbb{X}_8=\mathbb{S}_{\mu_2}$, $\mathbb{X}_9=\mathbb{S}_{\mu_3}$, and then build upon them one can compute
\begin{eqnarray*}
	&& \hspace{1.5cm}\mathbb{X}_3=\overline{\mathcal{R}^M(\overline{\mathbb{X}_5}, [2, 10])},\\
	&& \hspace{1.5cm}\mathbb{X}_1=\mathcal{R}^m(\mathbb{X}_3, [0, 16]),\\
	&& \hspace{1.5cm}\mathbb{X}_6=\overline{\mathcal{R}^m(\overline{\mathbb{X}_8}, [0, 10])},\\
	&& \hspace{1.5cm}\mathbb{X}_7=\mathcal{R}^M(\mathbb{X}_9, [5, 10]),\\
	&& \hspace{1.5cm}\mathbb{X}_4=\{x: h_{\mathbb{X}_{6}}(x)\ge 0 \wedge h_{\mathbb{X}_{7}}(x)\ge 0\},\\
	&& \hspace{1.5cm}\mathbb{X}_2=\mathcal{R}^M(\mathbb{X}_{4}, [0, 15]),\\
	&& \hspace{1.5cm}\mathbb{X}_0=\{x: h_{\mathbb{X}_{1}}(x)\ge 0 \vee h_{\mathbb{X}_{2}}(x)\ge 0\}.
\end{eqnarray*}
A complete path $\bm{p}$ of an STLT is defined as a path that originates at the root node and terminates at a leaf node.
the STLT $\mathcal{T}_{\hat\varphi}$ (see Fig. \ref{Fig:1}) has     3 complete paths, i.e., 
\begin{align*}
    \bm{p}_1&=\mathbb{X}_0\vee \mathbb{X}_1 \mathsf{G}_{[0, 16]}\mathbb{X}_3\mathsf{F}_{[2, 10]}\mathbb{X}_5, \\
    \bm{p}_2&=\mathbb{X}_0\vee \mathbb{X}_2 \mathsf{F}_{[0, 14]}\mathbb{X}_4\wedge\mathbb{X}_6 \mathsf{G}_{[0, 10]}\mathbb{X}_8, \\
    \bm{p}_3&=\mathbb{X}_0\vee \mathbb{X}_2 \mathsf{F}_{[0, 14]}\mathbb{X}_4\wedge\mathbb{X}_7\mathsf{F}_{[5, 10]}\mathbb{X}_9,
\end{align*}
 
  A trajectory $(\bm{x}, t_k)\cong\mathcal{T}_{\hat\varphi}$ if and only if either of the condition (1) $(\bm{x}, t_k)\cong\bm{p}_1$ is satisfied or  (2) $(\bm{x}, t_k)\cong\bm{p}_2$ and $(\bm{x}, t_k)\cong\bm{p}_3$ are satisfied. 
\end{example*}

\begin{figure} [t]
\centering	\includegraphics[width=0.4\textwidth]{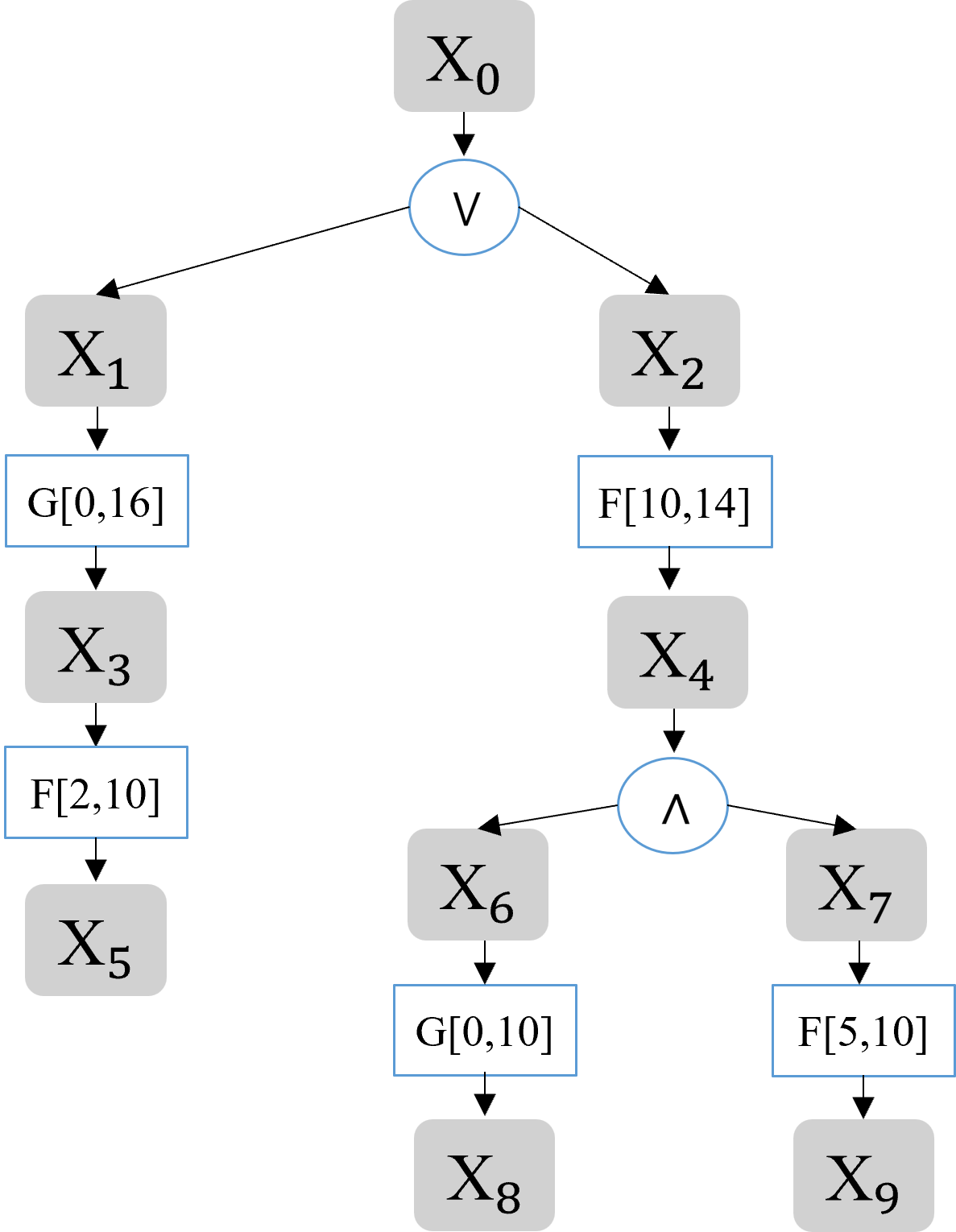}
	\caption{\footnotesize The STLT for the nested STL formula $\varphi=\mathsf{G}_{[0, 16]}\mathsf{F}_{[2, 10]}\mu_1 \vee \mathsf{F}_{[0, 14]}( \mu_2  \mathsf{U}_{[5, 10]}\mu_3)$, where $\mu_i, i=\{1,2,3\}$ are predicates. }
 \label{Fig:1}
\end{figure}
 
\section{Collaborative design of FD and FTC}

In this section, we propose an efficient framework for fault diagnosis (FD) and fault-tolerant control (FTC) in nonlinear systems under nested STL specifications, as illustrated in Fig. \ref{Fig:2}. \textbf{(Step 1)} The framework begins by introducing the Signal Temporal Logic Tree (STLT), a tool designed following the methodology in \cite{8} that transforms the satisfaction of an STL formula into a series of set invariance conditions.  \textbf{(Step 2)} A fault detection observer is constructed to leverage system dynamics for predicting future states during the evaluation of STL formula satisfaction, ensuring the system behavior aligns with desired specifications.  \textbf{(Step 3)} An online monitoring function is employed to assess the performance of the fault detection process in real-time.  \textbf{(Step 4)} Finally, a controller based on control barrier functions (CBFs) is designed. The STLT’s explicit representation of state constraints over specific time intervals provides clear guidelines for constructing the CBF, enabling the system to maintain state invariance and satisfy the STL specifications. The following subsection will introduce the theoretical results for the collaborative design process.

\begin{figure}
\centering	\includegraphics[width=0.4\textwidth]{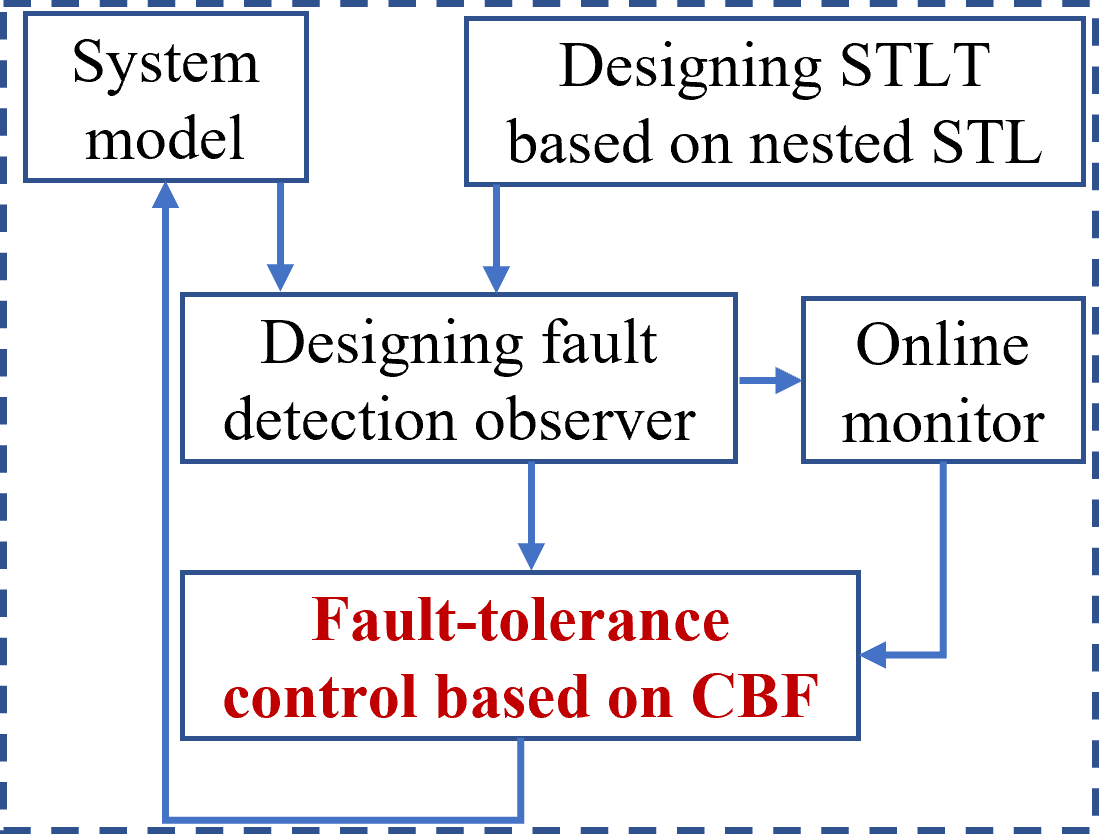}
	\caption{\footnotesize Framework of cooperative design of FD and FTC . }
 \label{Fig:2}
\end{figure}

\subsection{Fault detection observer with fault tolerant feasible sets}

In this work, a fragment of STL formulae is considered where the overall task is expressed as the conjunction of $N$ sub-formulae with nested temporal operators. Formally, an STL formula of the following form is considered
\begin{equation}
    \varphi = \bigwedge_{i=1,...,N} \varphi_i, \quad 
\end{equation}
Let $I = \{1, ... , N\}$ be the index set of sub-formulae. For each time instant $t_k$, we denote by $I_{t_k} = \{i \,|\, a_i \leq t_k \leq b_i\} \subseteq I$ the index set of formulae that are effective at time $t_k$. Similarly, it can define the index sets of sub-formulae effective before and after $t_k$ by $I_{< \,\, t_k} = \{i \,|\, t_k > b_i\}$ and $I_{> \,\, t_k} = \{i \,|\, t_k < a_i\}$, respectively. For each sub-formulae $i \in I$, it can denote by $O_i \in \{U, G, F\}$ the unique temporal operator in $\Phi_i$. We define $I^U_{t_k} = \{i \in I_t \,|\, O_i = U\}$; the same for $I_{t_k}^G$ and $I_{t_k}^F$.

As the system evolves, some sub-formulae may be satisfied or not effective anymore. Throughout the paper, notation $I' \subseteq I$ is used to denote the index set for those subformulae remaining unsatisfied. If $I' \cap I_{< \,\, t_k} = \emptyset$, the $I'$-remaining formula at instant $t_k$ is defined in the form of
\begin{equation}
    \hat{\varphi}_{t_k}^{I'} = \bigwedge_{i \in I' \cap I_{t_k}} \varphi_i^{[t_k,b_i]} \, \wedge \, \bigwedge_{i \in {I'}\cap I_{>t_k}} \varphi_i^{[a_i,b_i]}, \quad 
    \label{eqn:remainformula}
\end{equation}
where $[a_i,b_i]$ is time interval and $\varphi_i^{[t_k,b_i]}$ is attained from $\varphi_i^{[a_i,b_i]}$ by replacing the starting instant of the temporal operator from $a_i$ to $t_k$.

\begin{definition}
Given an STL formula $\varphi$ of form (\ref{eqn:remainformula}), the $I'$-remaining fault tolerant feasible set at instant $t_k$, denoted by $X_{t_k}^{I'}$, is  the set of states from which the system can satisfy the $I'$-remaining formula, i.e.,

\begin{equation}
\begin{aligned}
   X_{t_k}^{I'} = \left\{ x_{t_k} \in X \,\middle|\, \exists u_{t_k:T-1} \in U_{T-t_k}, \right. \\
   \left. \text{s.t.} \, x_{t_k} \xi_f(x_{t_k}, u_{t_k:T-1}) \models \hat{\varphi}_{t_k}^{I'} \right\}.
\end{aligned}
\end{equation}

\end{definition}

Otherwise, if $I' \cap I_{< \,\, t_k} \neq \emptyset$, it means that there are some sub-formulae that should have been satisfied. On this occasion, it can define $X_{t_k}^{I'} = \emptyset$ since the overall task is already failed.

Methods for computing the \( I' \)-remaining fault-tolerant feasible sets \( X_{t_k}^{I'} \) at time \( t_k \) is based on the definition of potential index set in \cite{10} as follows.  

\begin{definition}[Potential index set \cite{10}]  
At each \( t_k \in [0, T] \), a subset \( I' \subseteq I \) is a potential index set if:  
\begin{enumerate}  
    \item \( I_{<t_k} \cap I' = \emptyset \),  
    \item \( I_{>t_k} \subseteq I' \),  
    \item {\( i \in I_{t_k} \), then \( O_i = G \lor [O_i = U \lor O_i = F \land t_k = a_i]\subseteq I' \)}.  
\end{enumerate}  
\end{definition}  

The set of all potential index sets for \( t_k \) is denoted as \( \mathcal{I}_{t_k} \), and the corresponding feasible sets are \( X_{t_k} = \{ X_{t_k}^{I'} \mid I' \in \mathcal{I}_{t_k} \} \). Our goal is to compute \( X_{t_k} \) for all \( t_k \in [0, T] \). 

To compute \( X_{t_k} \),  a recursive backward approach is adopted. Given \( X_{t_k+1} \), each \( X_{t_k}^{I'} \) is determined using the elements in \( X_{t_k+1} \). Note that not all \( I'' \in \mathcal{I}_{t_k+1} \) are valid successors of \( I' \in \mathcal{I}_{t_k} \). The concept of successor sets is introduced to identify valid transitions.  

\begin{definition}[Successor sets \cite{10}]  
A set \( I'' \in \mathcal{I}_{t_k+1} \) is a successor of \( I' \in \mathcal{I}_{t_k} \) if:  
\[
\forall i \in I_{t_k+1} : (O_i = U \land i \notin I') \Rightarrow i \in I''.  
\]  
The set of all successors of \( I' \) is denoted as \( \text{succ}(I', t_k) \subseteq \mathcal{I}_{t_k+1} \).  
\end{definition}  

Note that the successor set of \( I' \) may not be unique in general since for those sub-formulae that have not yet been satisfied in \( I' \) and are still effective at the next instant, they can be either in \( I'' \) or not depending on the current state of the system. Inspired from \cite{10}, it can define the satisfaction sets and consistent regions. For a pair \( (I', I'') \) where \( I'' \in \text{succ}(I', t_k) \), the satisfaction sets and consistent regions are defined to facilitate the transition.  

\begin{definition}[Satisfaction sets and regions]  
For \( (I', I'') \), the satisfaction set is:  
\[
\text{sat}(I', I'') = \{ i \in I' : O_i = U \land i \notin I'' \}.  
\]  
The consistent region is:  
\[
H_{t_k}(I', I'') = \bigcap_{i \in I' \cap {I}_{t_k}} H_i,  
\]  
where:  
\[
H_i =  
\begin{cases}  
H_i^1 \cap H_i^2 & \text{if } i \in \text{sat}(I', I''), \\  
H_i^1 \setminus H_i^2 & \text{if } O_i = U \lor (O_i = F \land i \notin \text{sat}(I', I'')), \\  
H_i & \text{if } O_i = G.  
\end{cases}  
\]  
\end{definition}

Finaly, to compute \( X_{t_k}^{I'} \), two conditions must hold:  
\begin{itemize}  
    \item The state must remain in the consistent region \( H_{t_k}(I', I'') \) at \( t_k \).  
    \item It must transition to \( X_{t_k+1}^{I''} \) in one step to meet future requirements.  
\end{itemize}  

This framework ensures a systematic computation of feasible sets across all time steps and the following result can be obtained.

\begin{theorem}[fault tolerant feasible sets]
     Suppose that \( I' \) is the remaining index set and \( X_{t_k+1} \) is the set of all potential feasible sets at the next time instant. Then the \( I' \)-remaining feasible set \( X_{t_k}^{I'} \) defined in Definition 2 for the time instant \( t_k \) can be computed as follows
\begin{equation}
      X_{t_k}^{I'} = \bigcup_{I'' \in \text{succ}(I', t_k)} (H_{t_k}(I', I'') \cap r(X_{t_k+1}^{I''})) 
      \label{eqn:ftfset}
\end{equation}

where $r(*)$ is the one-step set defined by: for any $s \subseteq X$
\[
    r({s}) = \{x \in {X} \mid \exists u \in \mathsf{U} \text{ s.t. } f(x, u) \in {s}\}
\]

\end{theorem}

\begin{proof} When $t_k = T$,  since succ($I'$, $t_k$) = {$\emptyset$}  and  sat($I'$, $I''$) = sat$_{\cup}$($I'$, $\emptyset$) = i $\in$ $I'$ : $O_i = U$,  from Eq. (7), the results have been obtained
\[
X^{I'}_{t_k} = \bigcap_{i \in \text{sat}(I', I'')} ({H}_{i}^{1} \cap {H}_{i}^{2}) \cap \bigcap_{i \in I' \backslash \text{sat}(I', I'')} {H}_{i},
\]
which is clearly the  $I'$-remaining feasible set of  $\hat{\Phi}^{I'}_{t_k} = \bigwedge_{i \in I'} \Phi_{i}^{[T, T]}$.
\begin{figure*}[ht]
\begin{equation}
\begin{split}
\hat{\Phi}_{t_k}^{I'} &= \bigwedge_{i\in I'\cap{I}_{t_k}}\Phi_i^{[t_k, b_i]}\wedge\bigwedge_{i\in{I}_{>t_k}}\Phi_i^{[a_i, b_i]} \\
&= \bigwedge_{i\in I'_{t_k}\setminus{I}_{t_k+1}}\Phi_i^{[t_k, t_k]}\wedge\bigwedge_{i\in I'_{t_k}\cap{I}_{t_k+1}}\Phi_i^{[t_k, b_i]}\wedge\bigwedge_{i\in{I}_{>t_k}}\Phi_i^{[a_i, b_i]} \\
&= \bigwedge_{i\in I'_{t_k}\setminus{I'}_{t_k+1}}\Phi_i^{[t_k, t_k]}\wedge\bigwedge_{i\in I'^{U}_{t_k}\cap{I}_{t_k+1}}\Phi_i^{[t_k, b_i]}\wedge\bigwedge_{i\in I'^{G}_{t_k}\cap{I}_{t_k+1}}\Phi_i^{[t_k, b_i]}\wedge\bigwedge_{i\in{I}_{>t_k}}\Phi_i^{[a_i, b_i]} \\
&= \bigvee_{\hat{I}\subseteq I'^{U}_{t_k}\cap{I}_{t_k+1}} \left(\underbrace{\bigwedge_{i\in I'_{t_k}\setminus{I}_{t_k+1}}\Phi_i^{[t_k, t_k]}\wedge\bigwedge_{i\in I'^{U}_{t_k}\cap{I}_{t_k+1}\setminus\hat{I}}\Phi_i^{[t_k, t_k]}\wedge\bigwedge_{i\in\hat{I}} G_{[t_k, t_k]} x\in{H}_i^1}_{\psi_1(\hat{I})}\right.\\
&\quad \left.\underbrace{\wedge\bigwedge_{i\in I'^{G}_{t_k}\cap{I}_{t_k+1}} (\Phi_i^{[t_k, t_k]}\wedge\Phi_i^{[t_k+1, b_i]})\wedge \bigwedge_{i\in\hat{I}} \Phi_i^{[t_k+1, b_i]}\wedge \bigwedge_{i\in{I}_{>t_k}}\Phi_i^{[a_i, b_i]} }_{\psi_2(\hat{I})}\right)
\end{split}
\label{eqn:formuladec}
\end{equation}
\end{figure*}

For the case where \( t_k \neq T \), the expression for ${\hat{\Phi}^{I'}}_{t_k}$ can be formulated in accordance with Eq.(\ref{eqn:formuladec}). Here, \( I'_{t_k} \) represents the intersection \( I' \cap {I}_{t_k} \), with \( {I'^G_{t_k}} \) and \( {I'^U_{t_k}} \) denoting the subsets of indices corresponding to the temporal operators ``Always'' and ``Until,'' respectively. Conceptually, Eq.(\ref{eqn:formuladec}) partitions \( I'_{t_k} \) into two distinct components: \( I'_{t_k} \setminus {I}_{t_k+1} \) and \( I'_{t_k} \cap {I}_{t_k+1} \), which refer to the indices of the final and non-final instants, respectively. 

Further, \( I'_{t_k} \cap {I}_{t_k+1} \) can be decomposed based on the specific temporal operators, as illustrated in the second line of Eq.~(\ref{eqn:formuladec}). In the case of ``Until'' sub-formulae, with indices belonging to \( I'^U_{t_k} \cap {I}_{t_k+1} \), the satisfaction can either occur at the current instant or be deferred to the next. It can define \( \hat{I} \) as the set of indices for the sub-formulae that are not satisfied at the current instant, which can be any subset of \( I'^U_{t_k} \cap {I}_{t_k+1} \).

Subsequently, the expression for $\hat{\Phi}^{I'}_{t_k}$ can be refined into the form given in the third line of Eq.~(\ref{eqn:formuladec}), wherein for each possible \( \hat{I} \), the corresponding formula is divided into two components: one that pertains solely to the current state, denoted by \( \psi_1(\hat{I}) \), and another that captures the future temporal requirements, denoted by \( \psi_2(\hat{I}) \).

It is observed that, for any $\hat{I}$, the results have been obtained $x_{t_k} \xi (x_{t_k}, u_{t_k:T-1}) \models \psi_1(\hat{I}) \wedge \psi_2(\hat{I})$ if the following holds:

\begin{itemize}
    \item $x_{t_k} = \psi_1(\hat{I})$,
    \item $x_{t_k+1} \xi (x_{t_k+1}, u_{t_k+1:T-1}) \models \psi_2(\hat{I})$,
\end{itemize}

The first condition holds iff $x_{t_k}$ stays in region $H_{t_k}(I', I'')$ with $I'' = {I}_{>{t_k}} \cup (I'^G_{t_k} \cap {I}_{t_k+1}) \cup \hat{I}$ and $sat(I', I'') = (I'^U_{t_k} \setminus {I}_{t_k+1}) \cup (I'^U_{t_k} \cap {I}_{t_k+1} \setminus \hat{I})$. The second condition holds iff $x_{t_k+1}$ is in $I''$-remaining feasible set $X_{k+1}^{I''}$. The third condition holds iff $x_{t_k} \in r(X_{t_k+1}^{I''})$. Therefore, $\psi(\hat{I}):=\psi_1(\hat{I}) \wedge \psi_2(\hat{I})$ holds iff $x_{t_k}$ is in $H_{t_k}(I', I'') \cap r(X_{t_k+1}^{I''})$. Finally, recall that $\hat{\Phi}_{t_k}^{I'}$ is the disjunction of all possible $\psi(\hat{I})$. This suffices to consider all possible $I'' \in succ(I', t_k)$. Therefore, the results have been obtained $X^{I'}_{t_k} = \bigcup_{I'' \in succ(I', {t_k})} (H_{t_k}(I', I'') \cap r(X_{t_k+1}^{I''}))$ which is the same as Eq.(\ref{eqn:ftfset}), i.e., the theorem is proved.  

\end{proof}

\begin{remark}
In general, the computation of feasible sets necessitates the application of over- or inner-approximation techniques. When feasible sets  \( X_i \) are determined using over-approximation methods, there is a risk of missed alarms, as the resulting sets may include states that are not truly feasible. Conversely, employing inner-approximation techniques may lead to false alarms, as the approximated sets might exclude some feasible states. However, in the context of safety-critical systems, it is preferable to utilize inner-approximations to minimize the risk of missed alarms, thereby ensuring a more conservative and reliable system operation.
\end{remark}

\subsection{Online monitoring}

During the monitoring process, the system's state is observed by a monitor, which dynamically evaluates the satisfaction of STL tasks online, based on the trajectory generated by the system. Specifically, given a STL formula $\varphi$ with time horizon $T$ and a (partial) signal $\mathbf{x}_{0:t} = x_0 x_1 \ldots x_{t_k}$ (also called \textit{prefix}) up to time instant $t_k < T$, it can say $\mathbf{x}_{0:t_k}$ is
\begin{itemize}
    \item \textit{Violated} if $\mathbf{x}_{0:t_k} \xi (\mathbf{x}_{t_k}, \mathbf{u}_{t_k:T-1}) \not\models \Phi$ for any $\mathbf{u}_{t_k:T-1}$;
    \item \textit{Feasible} if $\mathbf{x}_{0:t_K} \xi (\mathbf{x}_{t_k}, \mathbf{u}_{t_k:T-1}) \models \Phi$ for some $\mathbf{u}_{t_k:T-1}$.
\end{itemize}
Then an online monitor is a function
\[
\mathcal{M} : \mathcal{X}^* \to \{0, 1\}
\]
such that, for any prefix $\mathbf{x}_{0:t_k}$, the results have been obtained
\[
\mathcal{M}(\mathbf{x}_{0:t_k}) = 1 \iff \mathbf{x}_{0:t_k} \text{ is violated},
\]
where $\mathcal{X}^*$ denotes the set of all finite sequences over $\mathcal{X}$. In \cite{19}, an efficient algorithm has been introduced for the synthesis of model-based periodic monitors, encompassing both pre-computational offline processes and real-time online execution. In the preparatory offline phase, the algorithm leverages model data to calculate in advance all potential remaining feasible sets anticipated at each time point. Subsequently, during the online phase, the focus shifts to straightforwardly tracking the index set $I'$ that corresponds to the unresolved formulae, and the monitoring decision can be made by checking whether or not $x_t \in X_{t_k}^{I'}$. At this point, we need to implement fault tolerant control with the help of the control synthesis in the next section.

\subsection{CBF time interval encoding}

In this section, the design of the CBF for a given STL formula is summarized. The approach involves designing a time interval encoding and appropriate CBFs to enforce the system trajectory to satisfy the STLT. 
 Before constructing the CBFs, the starting time of each subformula must be defined. While the temporal operator ``Until" is part of the semantics, it is equivalent to a combination of the ``Eventually" and ``Always" operators. For an STL operator \( O \in \{\wedge, \vee, \mathsf{F}_{[a, b]}, \mathsf{G}_{[a, b]}, \mathsf{U}_{[a, b]}\} \), if \( \varphi \) includes the ``Until" operator, such as \( \varphi = \varphi_1 \mathsf{U}_{[a, b]} \varphi_2 \), it is encoded as \( \hat{\varphi} = \mathsf{G}_{[0, b]} \varphi_1 \wedge \mathsf{F}_{[a, b]} \varphi_2 \).

The possible start time (interval) of \( \Theta \), which determines when to evaluate the satisfaction of \( \varphi_1 O \varphi_2 \) or \( O \varphi \), is defined as:

\[
[\underline{t}(\Theta), \bar{t}(\Theta)] := 
\begin{cases}
[0, 0], & \text{if } \Theta \in \{\wedge, \vee\}, \\
[a, b], & \text{if } \Theta \in \{\mathsf{F}_{[a, b]}\}, \\
[a, a], & \text{if } \Theta \in \{\mathsf{G}_{[a, b]}\}.
\end{cases}
\]

For the logical operators \( \wedge \) and \( \vee \), the start time is 0. For \( \mathsf{G}_{[a,b]} \), the start time is \( a \). For \( \mathsf{F}_{[a,b]} \), any time in the interval \( [a, b] \) is possible, so the start time is set to \( [a, b] \).

Following \cite{10}, the duration of \( \Theta \) is defined as:

\[
\mathcal{D}(\Theta) := 
\begin{cases}
0, & \text{if } \Theta \in \{\wedge, \vee, \mathsf{F}_{[a, b]}\}, \\
b - a, & \text{if } \Theta \in \{\mathsf{G}_{[a, b]}\}.
\end{cases}
\]

The root node of \( \mathcal{T}_{\hat{\varphi}} \) is denoted \( \mathbb{X}_{\text{root}} \). Let \( \mathbb{X} \) be the set of all nodes in the STLT \( \mathcal{T}_{\hat{\varphi}} \). For a set node \( \mathbb{X}_i \in \mathbb{X} \), the possible start time (interval) and duration are given by \( [\underline{t}_s(\mathbb{X}_i), \bar{t}_s(\mathbb{X}_i)] \) and \( \mathcal{D}(\mathbb{X}_i) \), respectively. The parent of node \( \mathbb{X}_i \) is denoted \( \text{PA}(\mathbb{X}_i) \), where \( \text{PA}(\mathbb{X}_i) \) is an operator node, and \( \text{PA}(\text{PA}(\mathbb{X}_i)) \) is a set node.

The calculation of the start time (interval) for each set node $\mathbb{X}_i$ is needed for ensuring the satisfaction of the STLT $\mathcal{T}_{\hat\varphi}$. Algorithm 3 in \cite{8} provided a detailed method to obtain the start time. However, due to the uncertainty of the start time for the temporal operator, the start times of some set nodes calculated in \cite{8} may be unknown and belong to an interval. To address this issue, this work introduces the event-triggered scheme for updating the start times, as described in \cite{10}. For each set node $\mathbb{X}_i$ and $\underline{t}_{s}(\mathbb{X}_i)\neq \bar{t}_{s}(\mathbb{X}_i)$, an event is triggered at time $t_k$ if:
\begin{equation}\label{trigger_condition}
    {t_k} \in [\underline{t}_{s}(\mathbb{X}_i)), \bar{t}_{s}(\mathbb{X}_i))] \;\wedge \; \bm{x}(t_k)\in \mathbb{X}_i.
\end{equation}
Once an event is triggered, it can update the start times of the set nodes with algorithm 4 in \cite{8}. Note that a set node's start time is fixed once an event is triggered. 

After the start time is calculated, the STLT  guides the design of CBFs for a given STL formula \( \varphi \). The goal is to determine the time interval and design appropriate CBFs to ensure that the system trajectory satisfies the STLT. Before designing CBFs, we first define the concept of a temporal fragment.

\begin{definition}
   Given the STLT for a nested STL formula $\varphi$, it is necessary to design one CBF for each temporal fragment $f_i$. Denote by $f_i=\Theta_{f_i}\mathbb{X}_{f_i}$, where $\Theta_{f_i}$ and $\mathbb{X}_{f_i}$ are the temporal operator node and the set node contained in $f_i$. Note that $\mathbb{X}_{f_i}$ is represented by its value function $\mathbb{X}_{f_i} = \{x_{t_k}: h_{\mathbb{X}_{f_i}}(x_{t_k})\ge 0 \}$. 
\end{definition}
A temporal fragment $f_j$ is referred to as the \emph{predecessor} of another temporal fragment $f_i$ (or $f_i$ the \emph{successor} of $f_j$) if there exists a complete path $\bm p$ such that $\bm p = ... f_j \bm{p}^\prime f_i ...$ where  $\bm{p}^\prime$ does not contain any temporal fragments. The $f_i$ is termed a \emph{top-layer temporal fragment} if  $f_i$ has no predecessor temporal fragment. 

  Consider a differentiable function $\mathfrak{b_i} : X \times [t_0, t_1] \rightarrow \mathbb{R}$ and the associated set $\mathcal{C}(t_k) := \{ x \in X \mid \mathfrak{b_i}(x, t_k) \geq 0 \}.$ Referring to the content in \cite{8}, the corresponding CBF $\mathfrak{b}_i(x_{t_k}, t_k)$  is required to satisfy the following conditions:

\begin{itemize}
  \item[1)] $\mathfrak{b}_i(x_{t_k}, t_k)$ is continuously differentiable and is defined over $\mathcal{C}(t_k)\times[\min\{{t}_e(\text{PA}(\text{PA}(\mathbb{X}_{f_i}))), \underline{t}_s(\mathbb{X}_{f_i})\}, t_e(\mathbb{X}_{f_i})]$;
  
  \item[2)] $\mathfrak{b}_i(x_{t_k}, t_k) \leq h_{\mathbb{X}_{f_i}}(x), \forall t_k\in [\bar{t}_s(\mathbb{X}_{f_i}), t_e(\mathbb{X}_{f_i})]$,
\end{itemize}
where $t_e(\mathbb{X}_{i})=\bar{t}_s(\mathbb{X}_{i})+\mathcal{D}(\mathbb{X}_i)$ {(recall $\mathcal{D}(\mathbb{X}_i)$ is computed can be intrathecal as the end time of $\mathbb{X}_{i}$. Here, $\text{PA}(\cdot)$ denotes the parent node of a given node in the STLT. Specifically, for a set node $\mathbb{X}_i$, $\text{PA}(\mathbb{X}_{f_i})$ represents its parent operator node.

Define the \emph{time domain} of the CBF $\mathfrak{b}_i(x_{t_k},t_k)$ as
\begin{equation}\label{eq:timedomain}  [\underline{t}_{\mathfrak{b}_i},\bar{t}_{\mathfrak{b}_i}]:=[\min\{{t}_e(\text{PA}(\text{PA}(\mathbb{X}_{f_i}))), \underline{t}_s(\mathbb{X}_{f_i})\}, t_e(\mathbb{X}_{f_i})].
\end{equation}
This is to guarantee that the CBF $\mathfrak{b}_i$, which corresponds to the temporal fragment $f_i$, is activated at  ${t}_e(\text{PA}(\text{PA}(\mathbb{X}_{f_i})))$, for which the activation of the predecessor of $f_i$ ends, or at  $\underline{t}_s(\mathbb{X}_{f_i})$, for which $f_i$ becomes active at its earliest, whichever comes earlier. A formal statement on this is given in Lemma \ref{lem:time_sequence}.

\begin{lemma} \label{lem:time_sequence}
    Let $f_i$ be a non-top-layer temporal fragment, and $f_j$ be the predecessor of $f_i$ in the constructed STLT. Denote their respective CBFs $\mathfrak{b}_j(x, t), \mathfrak{b}_i(x, t)$. Then $ \underline{t}_{\mathfrak{b}_j} \leq  \underline{t}_{\mathfrak{b}_i} \leq \bar{t}_{\mathfrak{b}_j} \leq \bar{t}_{\mathfrak{b}_i}  $.
\end{lemma}
\begin{proof}
It can be deduced from the tree structure that the predecessor of a non-top-layer temporal fragment is unique. Denote the set nodes in the fragments $f_j$ and  $f_i$ are $\mathbb{X}_{f_j}, \mathbb{X}_{f_i}$, respectively. The inequalities can be obtained as follows: 1) in view of \eqref{eq:timedomain}, $\underline{t}_{\mathfrak{b}_j} \leq  {t}_{e}(\mathbb{X}_{f_j})$ and $\underline{t}_{\mathfrak{b}_j} \leq  \underline{t}_{s}(\mathbb{X}_{f_j}) \leq \underline{t}_{s}(\mathbb{X}_{f_i}) $, thus $\underline{t}_{\mathfrak{b}_j} \leq \underline{t}_{\mathfrak{b}_i} = \min({t}_{e}(\mathbb{X}_{f_j}), \underline{t}_{s}(\mathbb{X}_{f_i}) )$; 2) from \eqref{eq:timedomain}, $\underline{t}_{\mathfrak{b}_i} \leq t_{e}(\mathbb{X}_{f_j})=\bar{t}_{\mathfrak{b}_j} $; 3) and the definition of $t_e(\cdot)$, 
$\bar{t}_{\mathfrak{b}_j} =t_{e}(\mathbb{X}_{f_j}) = \bar{t}_s(\mathbb{X}_j) + \mathcal{D}(\mathbb{X}_j)\leq   \bar{t}_s(\mathbb{X}_i) + \mathcal{D}(\mathbb{X}_i) =  t_{e}(\mathbb{X}_{f_i}) =\bar{t}_{\mathfrak{b}_i}$.
\end{proof}

If $f_i$ is not a top-layer temporal fragment, then the third condition on the corresponding CBF $\mathfrak{b}_i(x_{t_k}, t_k)$ is
\begin{itemize}
  \item[3)] $  \mathfrak{b}_i(x_{t_k},\underline{t}_{\mathfrak{b}_i}) \ge 0, \forall x_{t_k}\in \{x_{t_k}: \mathfrak{b}_j(x_{t_k},\underline{t}_{\mathfrak{b}_i}) \ge 0\}$,
\end{itemize}
  where $f_j$ is the unique predecessor of $f_i$.
Note that $\mathfrak{b}_j(x_{t_k},\underline{t}_{\mathfrak{b}_i})$ is well-defined in view of Lemma \ref{lem:time_sequence}.

Once the CBF is obtained for the nested STL formula, then the control strategy is given by solving a quadratic program
\begin{align}
\min_{u \in U} \quad & u^T Q u \\
\text{s.t.}  & L_f \mathfrak b_i(x_{t_k}) + L_g \mathfrak b_i(x_{t_k})u \geq  -\delta \hspace{1pt} \mathfrak b_i(x_{t_k})
\end{align}
where $\delta$ is a locally Lipschitz continuous class $\mathbb{K}$ function.

\subsection{Model predictive control synthesis with CBF}
To ensure the existence of feasible solutions in the aforementioned dynamic optimization control problems, we introduce the concept of fault-tolerant control recursive feasibility (FTCRF). FTCRF denotes the ability of a control strategy to maintain the system's controllability recursively under fault conditions, ensuring global stability and adherence to constraints. For simplicity and consistency, the notation $t_k$ used earlier will be replaced by $t$ throughout this paper. 

Before formally defining FTCRF, we introduce the concept of periodic safety, which requires the system trajectory to periodically revisit a subset of a forward-invariant safe set. 

\begin{definition}\label{def:period_stab}
Given sets $\mathcal{X}_{T'}, \mathcal{X} \subset \mathbb{R}^{n_x}$, where $\mathcal{X}_{T'} \subset \mathcal{X}$, and a time period $T' > 0$, the set $\mathcal{X}_{T'}$ is periodically safe with respect to the safe set $\mathcal{X}$ for the closed-loop system \eqref{x0} if, for all $x(0) \in \mathcal{X}_{T'}$, the following holds:  
\begin{align} \label{eq:lowLevelCnstr}
    x(iT') \in \mathcal{X}_{T'}, \quad x(t) \in \mathcal{X}, \quad \forall i \in \mathbb{N}, \; \forall t \geq 0.
\end{align}
\end{definition}

Here, $\mathcal{X} \subset \mathbb{R}^{n_x}$ and $\mathcal{X}_{T'} \subset \mathcal{X} \ominus \mathcal{D}$, where $\mathcal{D} = \{x \;|\; \|x\| \leq d\}$ for some $d > 0$. The input constraint set is given by $\mathcal{U} = \{u \;|\; A_u u \leq b_u\}$, where $A_u \in \mathbb{R}^{m \times n_u}$ and $b_u \in \mathbb{R}^m$. The time constant $T'$ is a user-defined parameter specifying the update frequency of the planned trajectory. 

Inspired by \cite{12}, we present a hierarchical strategy comprising two levels. The high-level planner generates a reference trajectory $z(t)$, while a low-level controller tracks this trajectory to ensure that the closed-loop trajectory $x(t)$ satisfies the state constraints. The objectives aim to ensure system safety through the forward invariance of $\mathcal{X}$ and the \textit{periodic fixed-time stability} of $\mathcal{X}_{T'}$. Periodic fixed-time stability mandates that the system trajectory revisits $\mathcal{X}_{T'}$ at discrete times $iT'$, where $i \in \mathbb{R}_{0+}$.  

The concept of periodic safety establishes a critical link between the low-level controller and the high-level planner, ensuring adherence to tracking discrete waypoints within a defined region. This integration is vital for maintaining both efficiency and safety. Subsequently, we define the domain of attraction (DoA):

\begin{definition}[\textbf{Domain of Attraction (DoA)}]\label{def:FT-DoA}
Given a set $\mathcal{C} \subset \mathbb{R}^{n_x}$ and a time $T' > 0$, a set $D_{\mathcal{C}} \subset \mathbb{R}^{n_x}$ is the DoA of $\mathcal{C}$ for the closed-loop system \eqref{x0} if:  
\begin{itemize}
    \item[i)] For all $x(0) \in D_{\mathcal{C}}$, $x(t) \in D_{\mathcal{C}}$ for all $t \in [0, T')$, and  
    \item[ii)] There exists $0 \leq T'_{\mathcal{C}} \leq T'$ such that $\lim_{t \to T'_{\mathcal{C}}} x(t) \in \mathcal{C}$.  
\end{itemize}
\end{definition}

The DoA is critical when the input is constrained to $u \in \mathcal{U}$, as fixed-time convergence cannot be guaranteed for arbitrary initial conditions. To characterize the DoA, we introduce a class of barrier functions called fixed-time  barrier functions:  

\begin{definition}\label{def:FT-barrier}
A continuously differentiable function $\mathfrak{b}_i : \mathbb{R}^{n_x} \to \mathbb{R}$ is a fixed-time barrier function for the set $\mathcal{S} = \{x \;|\; \mathfrak{b}_i(x) \geq 0\}$ with time $T_{\mathcal{S}} > 0$ for the closed-loop system \eqref{x0} if there exist parameters $\delta \in \mathbb{R}$, $\alpha > 0$, $\gamma_1 = 1 + \frac{1}{\mu}$, and $\gamma_2 = 1 - \frac{1}{\mu}$ for some $\mu > 1$ such that:  
\begin{align}\label{eq:FTBarrier}
\dot{\mathfrak{b}}_i(x) \geq -\delta \mathfrak{b}_i(x) + \alpha \max\{0, -\mathfrak{b}_i(x)\}^{\gamma_1} + \alpha \max\{0, -\mathfrak{b}_i(x)\}^{\gamma_2},
\end{align}
for all $x \in D_{\mathcal{S}} \subset \mathbb{R}^{n_x}$, where $T_{\mathcal{S}}$ and $D_{\mathcal{S}}$ depend on $\frac{\delta}{2\alpha}$.
\end{definition}
\begin{lemma}
Using \eqref{eq:FTBarrier}, the set $D_{\mathcal{S}}$ is the DoA of $\mathcal{S}$ with time $T_{\mathcal{S}}$ if:  
{\small
\begin{align*}
    D_{\mathcal{S}} & = \begin{cases}
    \mathbb{R}^{n_x}, & r < 1, \\  
    \left\{x \;|\; \mathfrak{b}_i(x) \geq -k^\mu \left(r - \sqrt{r^2 - 1}\right)^\mu\right\}, & r \geq 1,  
    \end{cases} \\  
    T_{\mathcal{S}} & = \begin{cases}
    \frac{\mu \pi}{2 r}, & r \leq 0, \\  
    \frac{\mu k}{\alpha(1 - k)}, & r \geq 1,  
    \end{cases}
\end{align*}}
where $r = \frac{\delta}{2\alpha}$ and $0 < \texttt{r}, k < 1$. The existence of a barrier function $b$ implies: (i) forward invariance of $D_{\mathcal{S}}$ and (ii) convergence to $\mathcal{S}$ within $T_{\mathcal{S}}$.  
\end{lemma}
\begin{proof}
 The proof of this lemma follows directly from \cite{garg2021characterization}. As the reasoning is relatively straightforward, the detailed steps are omitted for brevity. 
\end{proof}
The hierarchical control input is defined as:
\begin{equation}
    u(t) = u_l(t) + u_m(t),
\end{equation}
where $u_l$ and $u_m$ are determined by the policy:  
\begin{equation}\label{eq:policy}
 \begin{cases} 
u_l(t) = \pi_l\big(x(t), u_m(t), i\big),~ \dot{u}_m(t) = 0, & t \in \mathcal{T}_i, \\  
u_l^+(t) = u_l(t),~u_m^+(t) = \pi_m\big(x^+(t)\big), & t / T' \in \mathbb{N},
\end{cases}
\end{equation}
where $\mathcal{T}_i = [(i-1)T', iT')$. Here, $\pi_m : \mathbb{R}^{n_x} \to \mathcal{U}_M \subset \mathcal{U}$ is the high-level planner’s policy generating a reference trajectory, and $\pi_l : \mathbb{R}^{n_x} \times \mathbb{R}^{n_u} \times \mathbb{N} \to \mathcal{U}$ is the low-level tracking policy. The constraint set $\mathcal{U}_M \subset \mathcal{U}$ governs the control authority reserved for each policy and is a design parameter.  

\subsubsection{High-level reference planning}\label{sec:midLayer}
\textbf{Reference Model:} In this section, an assumption is made that the reference trajectory $z(t)$ is generated using the following piecewise LTI model:
\begin{equation}\label{eq:referenceModel}
\begin{aligned}
     \begin{cases}
    \begin{matrix*}[l] \dot{z}(t) = A z(t)+ B u_m(t) \end{matrix*}, &  t \in \cup_{i=0}^\infty  (iT', (i+1)T')\\
    \begin{matrix*}[l] z^+(t) = \Delta_{z}(x^-(t)) \end{matrix*}, &  t \in \cup_{i=0}^\infty \{iT'\} \\
    \end{cases},
\end{aligned}
\end{equation}
where $T'$ from~\eqref{eq:lowLevelCnstr} is specified by the user and $z^-(t) = \lim_{\tau \nearrow t}z(\tau)$ and $z^+(t) = \lim_{\tau \searrow t}z(\tau)$ denote the right and left limits of the reference trajectory $z(t) \in \mathbb{R}^n$, which is assumed right continuous. The matrices $(A, B)$ are known and, in practice, may be computed by linearizing the system dynamics~\eqref{x0} about the equilibrium point, i.e., the origin.
Finally, the reference input $u_m(t) \in \mathbb{R}^d$ and the \textit{reset map} $\Delta_{z}$, which depends on the state of the nonlinear system~\eqref{x0}, are given by the higher layer as discussed next. 
 
\noindent\textbf{Model Predictive Control:}
A Model Predictive Controller (MPC) is designed to compute the high-level input $u_m(t)$ that defines the evolution of the reference trajectory in~\eqref{eq:referenceModel}, and to define the reset map $\Delta_{z}$ for the model~\eqref{eq:referenceModel}. The MPC problem is solved at $1/T'$ Hertz and therefore the reference high-level input is piecewise constant, i.e., $\dot u_m(t) =0~\forall t \in \mathcal{T}$ where $\mathcal T = \cup_{i=0}^\infty  (iT', (i+1)T')$. 
Initially, this study introduces the following discrete-time linear model:
\begin{equation}\label{eq:linearDiscreteSystem}
    z^d_{i+1} = \bar A z^d_i + \bar B v_i,
\end{equation}
where the transition matrices are $\bar A = e^{A T'} \text{ and } \bar B = \int_0^{T'} e^{A(T'-\eta)}B d\eta$. Now notice that, as the high-level input  $u_m$ is piecewise constant, if at time $t_i = iT'$ the state $z(iT')=z^+(iT')=z^d_i$ and $u_m(iT')=v_i$, then at time $t_{i+1} = (i+1)T'$. This study shows that \begin{equation}\label{eq:relDisCon}
    z^-((i+1)T')=z^d_{i+1}.
\end{equation}
Given the discrete-time model~\eqref{eq:linearDiscreteSystem} and the state of the nonlinear system~\eqref{x0}, $x(iT')$, the following finite-time optimal control problem is solved at time $t_i = iT' \in \mathcal{T}^c$:
\begin{subequations}\label{eq:ftocp}
\begin{align}
    \min_{\boldsymbol{v}_i, z_{i|i}^d} \quad &\sum_{k = i}^{i+N-1}\big( || z_{k|i}^d||_Q + ||v_{k|i}||_R \big) + ||z_{i+N|i}^d||_{Q_f} \\
    \text{s.t.} ~\quad & z^d_{k+1|i} = \bar A z^d_{k|i} + \bar B v^d_{k|i}\\
    & ||z^d_{k+1|i} - z^d_{k|i}||_2 \leq d-c \label{eq: zi zi+1 close}\\ 
    &  z^d_{k|i} \in \mathcal{X}_T' \ominus \mathcal{C}, ~ v^d_{k|i} \in \mathcal{U}_m \\
    &  z^d_{i|i} - x(iT') \in \mathcal{C} \label{eq: x z close}\\
    &  z^d_{i+N|i} \in \mathcal{X}_F ,\forall k = \{i, \ldots, i+N-1 \}
\end{align}
\end{subequations}
where $||p||_Q = p^\top Qp$ and $\mathcal C = \{x\; |\; \|x\|\leq c\}$ for some $0<c<d$ such that $\mathcal X_T\ominus\mathcal C \neq \emptyset$. Problem~\eqref{eq:ftocp} computes a sequence of open-loop actions $\boldsymbol{v}_i^d=[v^d_{i|i},\ldots,v^d_{i+N|i}]$ and an initial condition $z^d_{i|i}$ such that the predicted trajectory steers the system to the terminal set $\mathcal{X}_F\subset \mathcal X_{T'}$, while minimizing the cost and satisfying state and input constraints.
Let 
\begin{equation}\label{eq:mpcOpt}
        \boldsymbol{v}_i^{d,*}=[v^{d,*}_{i|i},\ldots,v^{d,*}_{i+N|i}], \quad \boldsymbol{z}_i^{d,*}=[z^{d,*}_{i|i},\ldots,z^{d,*}_{i+N|i}]
\end{equation} be the optimal solution of \eqref{eq:ftocp}, then the high-level policy is \begin{equation}\label{eq:midLevPolicy}
\begin{aligned}
    \pi_{m}(x(iT')) = \begin{cases}
    \begin{matrix*}[l] {u_m}(t) = v^{d,*}_{i|i} \end{matrix*} &  t = iT' \in \mathcal{T}^c\\
    \begin{matrix*}[l] \dot u_m(t) \!=\! 0 \end{matrix*} &  t \in \mathcal{T} \\
    \end{cases}
\end{aligned}
\end{equation}
Finally, The reset map is defined for~\eqref{eq:referenceModel} as follows:
\begin{equation}\label{eq:returnMap}
\begin{aligned}
\Delta_{z}(x(iT')) = z_{i|i}^{d,*}.
\end{aligned}
\end{equation}

\subsubsection{Low-level control synthesis}\label{sec: low level u}
The low-level policy is designed $\pi_l$.
Consider the system dynamics \eqref{x0} under the effect of the policy \eqref{eq:policy}:
\begin{align}\label{eq:closed_loop_system}
    \dot x(t) = f\big(x(t)\big) +g\big(x(t)\big)\big(u_l(t)+u_m(t)\big).
\end{align}
The sets are  $\mathcal D_i$ and $\mathcal C_i$  defined as
\begin{align}
   \mathcal D_{i} & \triangleq z^-(iT') \oplus \mathcal D = \{x\; |\; \|x-z^-(iT')\|\leq d\},\\
    \mathcal{C}_i & \triangleq z^-(iT') \oplus \mathcal C = \{x\; |\; \|x-z^-(iT')\|\leq c\}.
\end{align}
The results are presented in Section \ref{sec:properties} that $\mathcal C_i\subset \mathcal D_{i+1}$ (guaranteed by bound on the rate change of the reference trajectory $z(t)$ in \eqref{eq: zi zi+1 close}) along with $\mathcal D_i\subset\mathcal X$ and $\mathcal C_i\subset \mathcal X_{T'}$ (guaranteed by \eqref{eq: x z close}) guarantees that closed-loop trajectories meet the objectives in \eqref{eq:lowLevelCnstr}.
Under these considerations, the low-level control objective for $t\in \mathcal T_i = [(i-1)T', iT')$ is to design the policy $\pi_l$ such that the set $\mathcal D_i$ is DoA for the set $\mathcal{C}_i$.
To this end, for the time interval $\mathcal T_i$ with $i\in \mathbb{R}_{0+}$, consider the candidate barrier function $b_i:\mathbb R^{n_x}\rightarrow\mathbb R$ defined as {
\begin{align}\label{eq: FT barrier h_i}
   \mathfrak b_i(x(t)) = \frac{1}{2}c^2-\frac{1}{2}\|x(t)-z^-(iT')\|^2, \quad t\in \mathcal T_i.
\end{align}}
and define the following QP:{\small
\begin{subequations}\label{QP gen}
\begin{align}
\min_{u_l, \delta } \; &\frac{1}{2}u_l^2 + \frac{1}{2}  \delta ^2 + c\delta \\
    \textnormal{s.t.} \; &  A_u(u_l+u_m)  \leq  \; B_u, \label{C1 cont const}\\
    & L_f \mathfrak b_i(x) + L_g \mathfrak b_i(x)(u_m+u_l)  \geq  -\delta \hspace{1pt} \mathfrak b_i(x) \nonumber \\& \hspace{120pt}+\alpha \max\{0,-\mathfrak b_i(x)\}^{\gamma_1} \nonumber\\
    &\hspace{120pt} +\alpha \max\{0,-\mathfrak b_i(x)\}^{\gamma_2} \label{C2 stab const}
\end{align}
\end{subequations}}\normalsize
where $q>0$, $\delta$ represents the slack variable. and $u_m = \pi_m(x(i-1)T')$. The optimal solution of the QP is denoted \eqref{QP gen} as $(u_l^\star(x,u_m,i), \delta ^\star(x,u_m,i))$ and define the low-level policy as 
\begin{align}\label{eq:low_level_policy}
    \pi_l(x(t),u_m(t),i) = u_l^\star(x(t),u_m,i).
\end{align}
The constraint \eqref{C1 cont const} guarantees that $u = u_l+u_m\in \mathcal U$. The parameters $\mu, \alpha, \gamma_1, \gamma_2$ in \eqref{C2 stab const} are fixed, and are chosen as $\alpha  = \max\left\{\frac{\mu k}{(1-k)T'},\frac{\mu \pi}{T'\sqrt{1-\texttt{r}^2}}\right\}$, $\gamma_1 = 1+\frac{1}{\mu}$ and $\gamma_2 = 1-\frac{1}{\mu}$ with $\mu>1$ and $0<\texttt{r}, k<1$, so that the closed-loop trajectories reach the zero super-level set of the barrier function $b_i$ within the time step $T'$. 

\subsubsection{Closed-loop straint properties}\label{sec:properties}
In this section, the properties of the proposed hierarchical control architecture are shown. Consider the closed-loop system \eqref{eq:closed_loop_system} under the control input \eqref{eq:policy} with policies $\pi_m$ and $\pi_l$ defined in \eqref{eq:midLevPolicy} and \eqref{eq:low_level_policy}, respectively. Below,  three reasonable assumptions for the analysis of the properties and the explanation of how the satisfaction of the closed-loop trajectories leads to the satisfaction of \eqref{eq:lowLevelCnstr} are provided.

\begin{assumption}\label{ass:lowLevel}
For each interval $\mathcal T_i$, the solution $\left(u^\star(x(t),u_m,i), \delta ^\star(x(t),u_m,i)\right)$ of the QP \eqref{QP gen} is continuous for all $t\in \mathcal T_i$ and the following holds
\begin{align}\label{eq:bar r}
    \sup_{t\in \mathcal T_i}\frac{\delta ^\star(x(t),u_m,i)}{2\alpha} \leq \bar r \triangleq \frac{(\frac{d^2-c^2}{2})^\frac{1}{\mu}}{2k} + \frac{k}{2(\frac{d^2-c^2}{2})^\frac{1}{\mu}},
\end{align}
where $\alpha = \max\{\frac{\mu k}{(1-k)T},\frac{\mu \pi}{T}\}$, $\gamma_1 = 1+\frac{1}{\mu}$ and $\gamma_2 = 1-\frac{1}{\mu}$ for some $\mu>1$ and $0<k<1$.
\end{assumption}
\begin{assumption}\label{ass:invariance}
The set $\mathcal{X}_F$ is invariant for the autonomous discrete time model $z^d((i+1)T') =\bar A z^d(iT')$ for all $i\in \mathbb{N}$. Furthermore, for all $i\in \mathbb N$, it holds that $||z^d(iT')-\bar Az^d(iT')|| \leq d - c$.
\end{assumption}
\begin{assumption}\label{assum: QP feas bd}
For all $x\in \partial \mathcal{C}_i$, $i\in \mathbb Z_+$, and $u_m\in \mathcal U_M$, there exists $u_l\in \mathcal U_l$ such that the following holds:
\begin{align*}
    L_fb_i(x) + L_gb_i(x)(u_m+u_l) \geq 0.
\end{align*}
\end{assumption}
\begin{remark}
For given input bounds (dictated by the set $\mathcal U)$, the value of the slack term $\delta$ in QP \eqref{QP gen} depends on the time of convergence $T$. Furthermore, the upper-bound in \eqref{eq:bar r} depends on the parameters $c$ and $k$, where $0<c<d$ is such that $\mathcal X_T\ominus\mathcal C$ is non-empty and $0<k<1$. Thus, in practice, numerical simulations can guide the choice of the parameters $c,k$, and the time $T'$, so that \eqref{eq:bar r} in Assumption \ref{ass:lowLevel} can be satisfied. Assumption \ref{ass:invariance}  is standard in the MPC literature~\cite{kouvaritakis2016model} and it allows us to guarantee that the MPC problem is feasible at all time instances. In practice, the set $\mathcal{X}_F$ can be chosen as a small neighborhood of the origin.
\end{remark}

It is shown that under the low-level controller, defined as the optimal solution of the QP in \eqref{QP gen}, the set $\mathcal D_i$ is an DoA for the set $\mathcal{C}_i$. To this end, the QP \eqref{QP gen} must be feasible for all $x$ so that the low-level controller is well-defined. The slack term $\delta$ ensures the feasibility of the QP in \eqref{QP gen} for all $x\notin \partial \mathcal{C}_{i}$ ($\partial \mathcal{C}_{i}$ is the boundary of $\mathcal{C}_{i}$). For the feasibility of the QP in \eqref{QP gen} for $x\in \partial \mathcal{C}_{i}$, the following assumption is made.

From Definition  \ref{def:FT-barrier}, it is known that DoA depends on the ratio $r=\frac{\delta}{2\alpha}$ and we have the following results. 

\begin{lemma}\label{prop:maximum value}
Consider the maximum value of $\delta^\star(x)$ as the solution of the QP \eqref{QP gen}, then $b_i$ is a finite-time barrier function for the set $\mathcal{C}_i$, and $\mathcal{D}_i$ serves as its DoA.
\end{lemma}
\begin{proof}
    Consider the fixed-time barrier function of Definition \ref{def:FT-barrier}, and define the maximum value of \( b^\star := \left(\frac{\alpha_{2}}{\alpha_{1}}\right)^{\frac{\mu}{2}} \) and the maximum value of \( \delta^{\star} := \alpha_{1}\left(b^{\star}\right)^{\frac{1}{\mu}} + \alpha_{2}\left(b^{\star}\right)^{-\frac{1}{\mu}} = 2\sqrt{\alpha_{1}\alpha_{2}} \) so that \( \alpha_{1} b(x)^{\frac{1}{\mu}} + \alpha_{2} b(x)^{-\frac{1}{\mu}} \geq \delta^{\star} \) for all \( x \in \mathbb{R}^{n} \), where $\alpha=\sqrt{\alpha_{1}\alpha_{2}}>0$. Thus, for \( r < 1 \), it holds that \( \delta_{1} < \delta^{\star} \leq \alpha_{1} b(x)^{\frac{1}{\mu}} + \alpha_{2} b(x)^{-\frac{1}{\mu}} \) for all \( x \), and so, \( \dot{b}(x) < 0 \) for all \( x \in \mathbb{R}^{n} \backslash \{0\} \). Since \( D_i \) is defined as the largest sub-level set of \( b \) such that \( \dot{b}(x) \) takes negative values for \( x \in D_i \backslash \{0\} \), it holds that in the case when \( 0 \leq r < 1 \), \( D_i = \mathbb{R}^{n} \), the origin is DoA, and other cases can refer to the Lemma 3 in  \cite{garg2021characterization}.
\end{proof}

\begin{lemma}\label{prop:mpcFeasibilityIMplications}
If the MPC problem~\eqref{eq:ftocp} is feasible at time $t_i=iT'$, then $x(iT') \in \mathcal D_{i+1}$.
\end{lemma}
\begin{proof}
By assumption of the lemma, the MPC problem \eqref{eq:ftocp} is feasible at time $t_i = iT'$. Now consider the optimal MPC solution~\eqref{eq:mpcOpt} at time $t_i = iT'$. By definition, It is observed that $x(iT') - z_{i|i}^{*,d} \in \mathcal{C}$, which implies that $||x(iT') - z_{i|i}^{*,d}|| \leq c$. Furthermore, by feasibility of the optimal MPC solution~\eqref{eq:mpcOpt} for problem~\eqref{eq:ftocp}, It is observed that 
$||z_{i+1|i}^{*,d}-z_{i|i}^{*,d}|| \leq d-c$. This implies that 
\begin{equation*}
\begin{aligned}
    ||x(iT') - z_{i+1|i}^{*,d}|| &= ||x(iT') -z_{i|i}^{*,d}+z_{i|i}^{*,d}- z_{i+1|i}^{*,d}|| \\
    &\leq ||x(iT') -z_{i|i}^{*,d}||+||z_{i|i}^{*,d}- z_{i+1|i}^{*,d}||\leq d
\end{aligned}
\end{equation*}
Finally, from~\eqref{eq:relDisCon} It is observed that $z_{i+1|i}^{*,d} = z^-((i+1)T')$. Thus, from the above equation,the results indicate that $||x(iT') - z^-((i+1)T')|| \leq d$, which implies that $x(iT') \in \mathcal D_{i+1}$.
\end{proof}

\subsubsection{MPC recursive feasibility and closed-Loop constraint satisfaction}
So far, It is observed that feasibility of the MPC guarantees that $x(iT')\in \mathcal D_{i+1}$, which, under the low-level control policy \eqref{eq:low_level_policy}, guarantees that $x((i+1)T')\in \mathcal C_{i+1}$. Thus, what is remaining to be shown is that the MPC \eqref{eq:ftocp} is recursively feasible, i.e., if \eqref{eq:ftocp} is feasible at $t = 0T$, then it is feasible at $t = iT'$ for all $i\in \mathbb N$. This would guarantee that $x(iT')\in \mathcal C_i$ (and hence, $x(iT')\in \mathcal X_{T'}$) for all $i\in \mathbb N$.

Now, the main result is ready to be stated, showing that the hierarchical control strategy in Section 3.4 leads to the satisfaction of \eqref{eq:lowLevelCnstr}.

\begin{theorem}\label{th:main_result}
\textbf{(Fault-tolerant control recursive feasibility):} Let Assumptions ~\ref{ass:invariance} -\ref{assum: QP feas bd} hold and consider the closed-loop system~\eqref{eq:closed_loop_system} under the control policy~\eqref{eq:policy}, where $\pi_m$ is defined in \eqref{eq:midLevPolicy} and $\pi_l$ is defined in \eqref{eq:low_level_policy}. If at time $t=0$ problem~\eqref{eq:ftocp} is feasible, then 
the closed-loop trajectories under the control policy \eqref{eq:policy} satisfy \eqref{eq:lowLevelCnstr}, i.e., the set $\mathcal X_{T'}$ is periodically safe w.r.t. the set $\mathcal X$ with period $T'$.
\end{theorem}
\begin{proof}
The proof proceeds by induction. First, it is shown that if at time $t_i = iT'$ the MPC problem~\eqref{eq:ftocp} is feasible, then at time $t_{i+1}=(i+1)T'$ the MPC problem~\eqref{eq:ftocp} is feasible. Let 
\begin{equation*}
    [z_{i|i}^{d,*}, z_{i+1|i}^{d,*},\ldots,z_{i+N|i}^{d,*}] \text{ and } [u_{i|i}^{d,*},\ldots,u_{i+N-1|i}^{d,*}]
\end{equation*}
be the optimal state input sequence to the MPC problem~\eqref{eq:ftocp} at time $t_i = iT'$. Then from the feasibility of the MPC problem and Proposition~\ref{prop:mpcFeasibilityIMplications} it is shown that $x(iT') \in \mathcal D_{i+1}$, 
\begin{equation*}
    x((i+1)T') \in \mathcal{C}_{i+1} =\{x ~| ~ || x-z^-((i+1)T')||\leq c\}.
\end{equation*}
Now notice that from equation~\eqref{eq:relDisCon}, it has been found that $z_{i+1|i}^{*,d} = z^-((i+1)T')$, which in turn implies that 
\begin{equation}\label{eq:discContConn}
    x((i+1)T')- z_{i+1|i}^{d,*} = x((i+1)T')- z^-((i+1)T') \in \mathcal{C}
\end{equation}
and therefore, by Assumption~\ref{ass:invariance}, the following sequences of states and inputs
\begin{equation}\label{eq:feasTr}
   \hspace{-5pt} [z_{i+1|i}^{d,*},\ldots,z_{i+N|i}^{d,*}, \bar A z_{i+N|i}^{d,*}],\;  [u_{i+1|i}^{d,*},\ldots,u_{i+N-1|i}^{d,*}, 0]
\end{equation}
are feasible at time $t_{i+1}=(i+1)T'$ for the MPC problem~\eqref{eq:ftocp}. It has been shown that if the MPC problem~\eqref{eq:ftocp} is feasible at time $t_i = iT'$, then the MPC problem is feasible at time $t_{i+1}=(i+1)T'$. Per assumption of the theorem, problem~\eqref{eq:ftocp} is feasible at time $t_0 =0$, and hence, it is concluded that the MPC problem~\eqref{eq:ftocp} is feasible for all $t_i = iT'$ and for all $i \in \mathbb{R}_{0+}$.

Next, it is shown that the feasibility of the MPC problem implies that state and input constraints are satisfied for the closed-loop system. Notice that by definition $u_m(t) = v_{i|i}^* \in \mathcal{U}_m $ for all $t \in [iT', (i+1)T')$, the low-level controller returns a feasible control action $u_l(t)$, thus it follows that
\begin{equation}
    u(t) = u_l(t) + u_m(t) \in \mathcal{U}, \forall t \in \mathbb{R}_{0+}.
\end{equation}
Finally, from the feasibility of the state-input sequences in~\eqref{eq:feasTr} for the MPC problem~\eqref{eq:ftocp},It is observed that
\begin{equation}\label{eq:feasSolCnstr}
    \begin{aligned}
    & x_{i|i}^{d,*} \in \mathcal{X}_{T'} \ominus \mathcal{C} \text{ and } x(iT') - x_{i|i}^{d,*} \in \mathcal{C},~\forall i\in\mathbb{R}_{0+}.
    \end{aligned}
\end{equation}
From the above equation, it is concluded that $x(iT') \in \mathcal{X}_{T'}$ for all $i\in\mathbb{R}_{0+}$. Note that since $z^-(iT')\in \mathcal X_{T'}\ominus\mathcal C$, $\mathcal C\subset \mathcal D$, and $\mathcal X_{T'} = \mathcal X\ominus\mathcal D$, it follows that $\mathcal D_i = \{z^-(iT')\}\oplus \mathcal D\subset \mathcal X$ for all $i\in \mathbb N$. From \cite[Theorem 1]{garg2021characterization}, the set $\mathcal D_i$ is forward-invariant for the closed-loop trajectories $x(t)$, i.e., $x(t)\in \mathcal D_i$ for $t\in [(i-1)T', iT')$ for all $i\in \mathbb N$. Hence, it follows that $x(t)\in \mathcal X$ for all $t\geq 0$. Thus, the closed-loop trajectories under the control policy \eqref{eq:policy} satisfy \eqref{eq:lowLevelCnstr}, i.e., the set $\mathcal X_{T'}$ is periodically safe w.r.t. the set $\mathcal X$ with period $T'$.
\end{proof}

\begin{theorem}
Consider a dynamical system (1) and a nested STL specification $\varphi$. According to section 3.2, if $\mathcal{M}(\mathbf{x}_{0:t})=1$, the control strategy is governed by the CBF coding protocol. If the initial condition $x_0 \in {X}_{\text{root}}$ and the condition \eqref{QP gen} is feasible, then the resulting system trajectory $(x_{t_k}, 0) \models \varphi$.
\end{theorem}
\begin{proof} Given a complete path $\bm p$ and an initial condition $x_0$, let $f_0, f_1, ..., f_N$ be the sequence of temporal fragments contained in $\bm p$ and $\mathfrak{b}_0, \mathfrak{b}_1, \ldots, \mathfrak{b}_N$ the corresponding CBFs. Assume that each $\mathfrak{b}_i, i\in 0, \ldots, N$ satisfies the conditions 1)-3). Furthermore, if $\mathfrak{b}_0(x_0,0)\ge 0$ and each of the CBFs $\mathfrak{b}_i$ satisfies the condition (35) during the corresponding time domain, then the resulting trajectory satisfies this complete path $\bm p$.

 Without loss of generality, assume that $f_i$ is the predecessor of $f_{i+1}$, $i=0,1,...,N-1$. For the top-level temporal fragment $f_0$, since $\mathfrak{b}_0(x_0,0)\ge 0$ and the CBF condition (19) holds in $[0,\bar{t}_{\mathfrak{b}_0}]$, it has been established that $ \mathfrak{b}_0(\bm{x_{t_k}},t_k) \ge 0, \forall t_k\in [0,\bar{t}_{\mathfrak{{b}}_0}]$. Now assume $\mathfrak{b}_i(\bm{x_{t_k}},t_k) \ge 0,  \forall t_k\in [\underline{t}_{\mathfrak{b}_i},\bar{t}_{\mathfrak{b}_i}]$. From condition 3), $\mathfrak{b}_{i+1}(\bm{x}(\underline{t}_{\mathfrak{b}_{i+1}}),\underline{t}_{\mathfrak{b}_{i+1}}) \ge  0$. In addition, the CBF of $\mathfrak{b}_{i+1}$ is satisfied for  $\forall t_k\in [\underline{t}_{\mathfrak{b}_{i+1}},\bar{t}_{\mathfrak{{b}_{i+1}}}]$, and then   $ \mathfrak{b}_{i+1}(\bm{x}(t),t) \ge 0, \forall t_k\in [\underline{t}_{\mathfrak{b}_{i+1}},\bar{t}_{\mathfrak{{b}_{i+1}}}]$. Inductively, the results are obtained $\mathfrak{b}_{i}(\bm{x}(t),t) \ge 0, \forall t_k \in [\underline{t}_{\mathfrak{b}_{i}},\bar{t}_{\mathfrak{{b}_{i}}}]$ for $i = 0,1,2, ..., N$.

In addition, $ \mathfrak{b}_{i}(\bm{x_{t_k}},t_k) \ge 0, \forall t_k \in [\underline{t}_{\mathfrak{b}_{i}},\bar{t}_{\mathfrak{{b}_{i}}}]$ implies that $\bm{x}(t)\in \mathbb{X}_i, \forall t_k\in [\bar{t}_s(\mathbb{X}_{f_i}), {t}_e(\mathbb{X}_{f_i})]$ from condition 2). One verifies that $[\bar{t}_s(\mathbb{X}_{f_i}), {t}_e(\mathbb{X}_{f_i})], \forall f_i$  is a valid time interval coding of the complete path from Definition \ref{Def:PathSaf} items i-iii). Thus, the resulting trajectory satisfies the complete path $\bm p$. The theorem is proved.
\end{proof}
\subsection{The overall algorithm}
In this subsection, we summarize the specific steps of the proposed CoD algorithm as shown in Algorithm 1.This algorithm aims to provide an integrated fault diagnosis and fault-tolerance control framework for nonlinear systems to meet nested STL specifications. The core idea of the algorithm is to use the dynamic characteristics of the system to predict future states and thus evaluate whether the STL formulas are satisfied. Here are the detailed steps of the algorithm:
\begin{algorithm}[H]
\caption{CoD of FD and FTC}
\label{alg:CoD_FD_FTC}
\begin{itemize}
    \item \textbf{Step 1.} Design control law.
    \begin{itemize}
        \item 1-1. Design STLT based on nested STL.
    \end{itemize}
    \item \textbf{Step 2.} Design fault detection observer.
    \begin{itemize}
        \item 2-1. Design fault tolerant feasible sets.
        \item 2-2. Design the output of observer.
    \end{itemize}
    \item \textbf{Step 3.} Design quadratic program.
    \begin{itemize}
        \item 3-1. Encode time for STLT.   
        \item 3-2. Design CBF for each temporal fragment. 
        \item 3-3. Design fault-tolerant control recursive feasibility for control strategy.
        \item 3-4. Solve quadratic program.
    \end{itemize}
\end{itemize}
\end{algorithm}
Fault Detection Observer Design (FD Observer Design):Initially, a fault detection observer is designed. The pumathcal Xose of this observer is to use the system's dynamic characteristics to predict future states, thereby assessing the satisfaction of STL formulas. This step is the core of fault detection (FD), aiming to verify whether the system's behavior complies with the desired specifications.

Evaluation of FD Performance (Evaluation of FD):After designing the fault detection observer, a function is constructed to evaluate the performance of fault detection. This step ensures that the fault detection mechanism can effectively identify and respond to system failures.

Controller Design (Controller Design):Based on the system model, a controller is designed using Control Barrier Functions (CBFs). The core of the controller design is to ensure that the system state remains within certain sets during specific time intervals, providing explicit guidelines for the CBF design.

Control Strategy Synthesis (Control synthesis):To address the challenge of ensuring a feasible solution in dynamic optimization control problems, the concept of fault-tolerant control recursive feasibility is introduced. The input design is transformed into a control strategy by solving a Quadratic Programming (QP) problem. This step translates the controller design into specific control inputs to achieve stable operation and performance optimization of the system.

\section{SIMULATION STUDY}
In this section, we apply the proposed fault diagnosis and fault-tolerant control co-design approach to two case studies with an integrator dynamic model and a nonholonomic mobile robots modeled by kinematic unicycles.
\begin{figure}[htbp]
\centering	\includegraphics[width=0.5\textwidth]{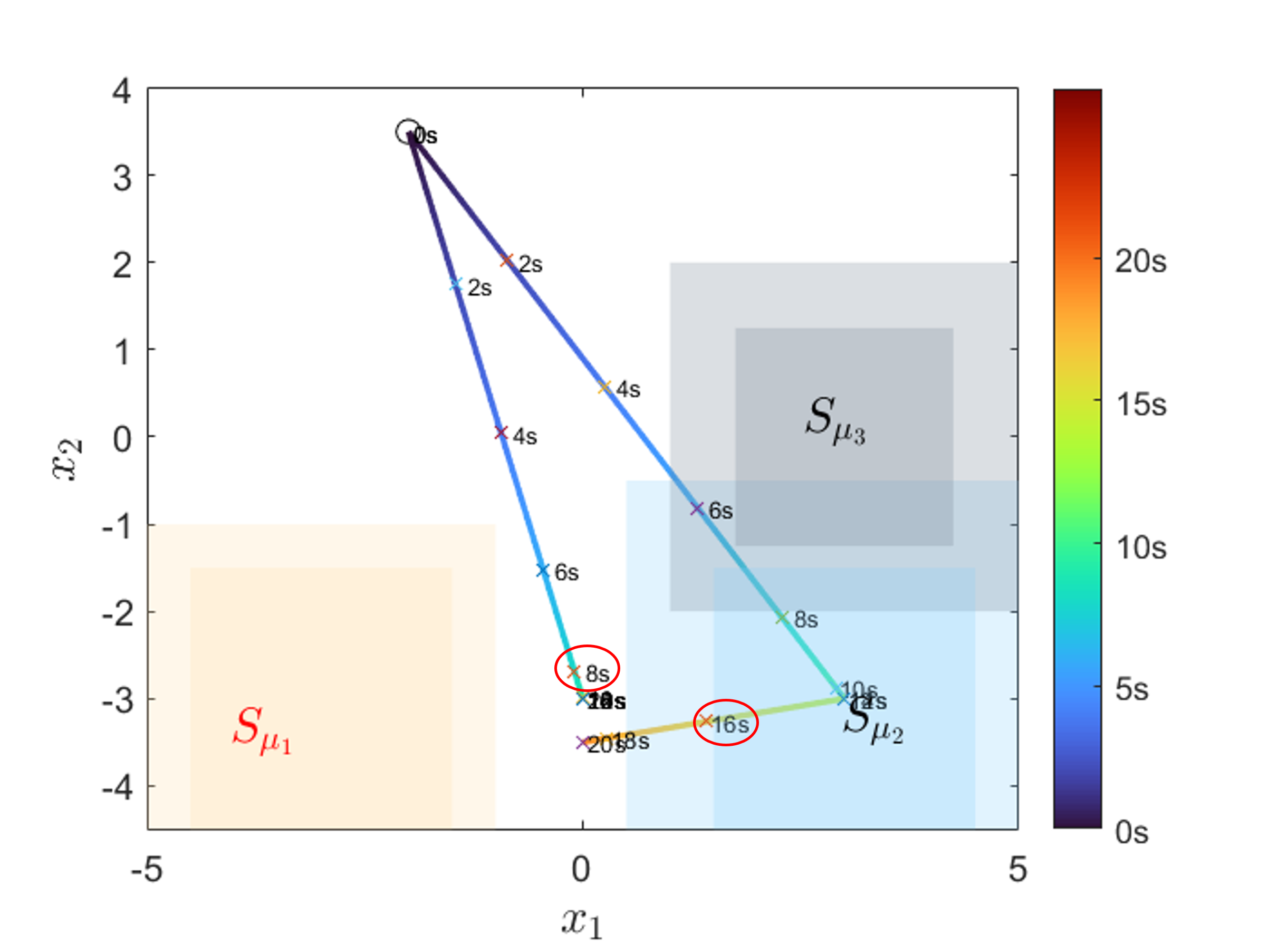}
	\caption{\footnotesize Two trajectories for the single integrator dynamics system. The yellow, blue, and gray regions denote the target regions \( S_{\mu_1} \), \( S_{\mu_2} \), and \( S_{\mu_3} \), respectively, while the light yellow, light blue, and light gray regions represent their corresponding fault-tolerant feasible sets \( X'_{t_k} \). The red circle marks a state \( x_{t_k} \notin X'_{t_k} \), indicating that the monitor can predict task failure in advance.}
 \label{Fig:3}
\end{figure}
\subsection{Integrator model}
We use an integrator model with a sample period of 1 s to model the dynamics of the vehicle

\begin{equation}
\bar{x}_{k+1} = \begin{bmatrix} 1 & 0 \\ 0 & 1 \end{bmatrix} x_k + \begin{bmatrix} 1 & 0 \\ 0 & 1 \end{bmatrix} u_k,
\end{equation}
where state is $x_k=[x_1, x_2]$, and control input is $u_k=[u_x, u_y]$, the control input set $u_k=\{u: |u_x|\le 1, |u_y|\leq 1\}$, and the  working space is $\mathbb{S}_{\mu_1}=\{x\in \mathbb{R}^2\mid x_1 \in [-4.5\ {-1.5}] \wedge x_2 \in [-4.5\ {-1.5}] \}$,  $\mathbb{S}_{\mu_2}=\{x\in \mathbb{R}^2\mid x_1 \in [1.5\ {4.5}] \wedge x_2 \in [-4.5\ {-1.5}] \}$, and $\mathbb{S}_{\mu_3}=\{x\in \mathbb{R}^2\mid x_1 \in [1.75\ {4.25}] \wedge x_2 \in [-1.25\ {1.25}] \}$. The STL task specification is given by $\varphi=\mathsf{G}_{[0, 16]}\mathsf{F}_{[2, 10]}\mu_1 \vee \mathsf{F}_{[10, 14]}( \mu_2  \mathsf{U}_{[5, 10]}\mu_3)$. In this simulation, we run the high-level MPC planner at 5Hz and the low-level controller at 10kHz with parameters d = 0.6, c = 0.005 and T'=0.2.

An important observation is that, for integrator dynamics and a given set node $\mathbb{X}{\varphi_1}$, the sets $\mathcal{R}^M(\mathbb{X}_{\varphi_1}, [a, b])$ and $\overline{\mathcal{R}^m(\overline{\mathbb{X}_{\varphi_1}}, [a, b])}$, which represent the set nodes obtained via the temporal operators $\mathsf{F}{[a, b]}$ and $\mathsf{G}_{[a, b]}$, respectively, exhibit monotonic growth with respect to the input set $u_k$. Consequently, it follows that
\begin{equation}
    \begin{aligned}
      & \mathbb{X}_5=\mathbb{S}_{\mu_1}, \mathbb{X}_8=\mathbb{S}_{\mu_2}, \mathbb{X}_9=\mathbb{S}_{\mu_3},\\
   & \mathbb{X}_4=\{x\in \mathbb{R}^2\mid x_1 \in [1.5\ {4.5}] \wedge x_2 \in [-4.5\ {-1.5}] \},\\
  & \mathbb{X}_3=\{x\in \mathbb{R}^2\mid x_1 \in [-14.5\ {8.5}] \wedge x_2 \in [-14.5\ {8.5}] \},\\
  & \mathbb{X}_2=\{x\in \mathbb{R}^2\mid x_1 \in [-12.5\ {18.5}] \wedge x_2 \in [-18.5\ {12.5}] \}, \\
  & \mathbb{X}_1=\{x\in \mathbb{R}^2\mid x_1 \in [-14.5\ {8.5}] \wedge x_2 \in [-14.5\ {8.5}] \},\\
  & \mathbb{X}_0=\{x\in \mathbb{R}^2\mid x_1 \in [-12.5\ {18.5}] \wedge x_2 \in [-18.5\ {12.5}]  \text{ or }  \\
  & \hspace{1cm} x\in \mathbb{R}^2\mid x_1 \in [-14.5\ {8.5}] \wedge x_2 \in  [-14.5\ {8.5}]\}. \\
    \end{aligned}
\end{equation}
Here $\mathbb{X}_0, ...,\mathbb{X}_5$ are subsets of what one could obtain with the input set $u_k$.
Let the temporal fragments be denoted as $f_1=\mathsf{G}_{[0, 16]} \mathbb{X}_3, f_2=\mathsf{F}_{[2, 10]} \mathbb{X}_{5}, f_3=\mathsf{F}_{[0, 14]} \mathbb{X}_4, f_4=\mathsf{G}_{[0, 10]} \mathbb{X}_{8}, f_5=\mathsf{F}_{[5, 10]} \mathbb{X}_{9} $, with their corresponding control barrier functions $\mathfrak{b}_1, \dots, \mathfrak{b}_5$.
\noindent
\begin{figure}
\centering	\includegraphics[width=0.5\textwidth]{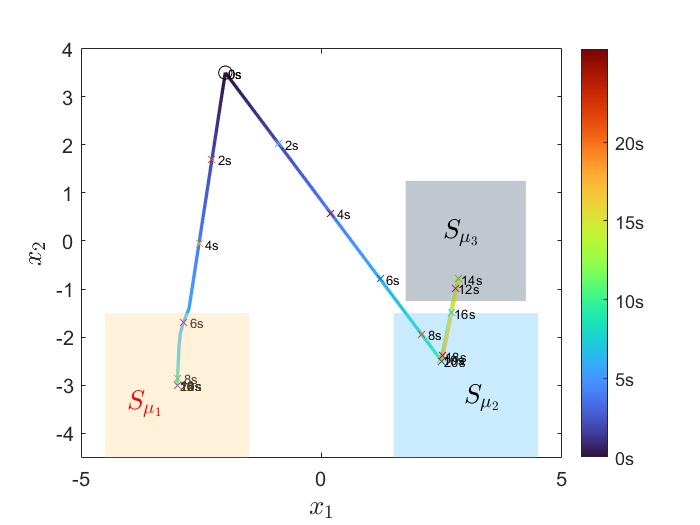}
	\caption{\footnotesize Two trajectories of a mobile robot with single integrator dynamics are synthesized using the proposed method under the STL specification $\varphi=\mathsf{G}_{[0, 16]}\mathsf{F}_{[2, 10]}\mu_1 \vee \mathsf{F}_{[10, 14]}( \mu_2  \mathsf{U}_{[5, 10]}\mu_3)$. The yellow, blue, and gray regions denote the target regions \( S_{\mu_1} \), \( S_{\mu_2} \), and \( S_{\mu_3} \). }
 \label{Fig:5}
\end{figure}
Before starting the online monitoring process, we first compute the fault tolerant feasible sets by theorem 2. The offline computation results are shown in Fig. 3. Specifically, the yellow, gray, and blue regions represent target regions $\mathbb{S}_{\mu_1}$, $\mathbb{S}_{\mu_2}$, and $\mathbb{S}_{\mu_3}$, respectively, while the larger areas $ X_{t_k}^{I'}$ in light yellow, light gray, and light blue denote the corresponding fault tolerant feasible regions. For three complete paths in the example, since $x_{t_k} \notin X_{t_k}^{I'} $ (The red circle is shown in Figure 3),  we can conclude immediately that the formula will be violated inevitably since there exists no controller under which the STLT formula can be satisfied.

 Based on the system model, a controller is designed using CBFs. By applying equation \eqref{eq:timedomain}, one can derive the initial starting time interval, duration, and time domain of the associated CBFs:
\begin{itemize} 
\itemindent=-12pt
    \item $[\underline{t}_s(\mathbb{X}_{3}), \bar{t}_s(\mathbb{X}_{3})]=[0, 16], \mathcal{D}(\mathbb{X}_{3}) = 16,[\underline{t}_{\mathfrak{b}_1}, \bar{t}_{\mathfrak{{b}}_1}]=[0, 16]$;
     
	\item $[\underline{t}_s(\mathbb{X}_{5}), \bar{t}_s(\mathbb{X}_{5})]=[2, 18], \mathcal{D}(\mathbb{X}_{5}) = 8,[\underline{t}_{\mathfrak{b}_2},\bar{t}_{\mathfrak{{b}}_2}] = [2, 26]$;
 
	\item $[\underline{t}_s(\mathbb{X}_{4}), \bar{t}_s(\mathbb{X}_{4})]=[10, 14], \mathcal{D}(\mathbb{X}_{4}) = 0,[\underline{t}_{\mathfrak{b}_3},\bar{t}_{\mathfrak{{b}}_3}] = [10, 14]$;
 
	\item $[\underline{t}_s(\mathbb{X}_{8}), \bar{t}_s(\mathbb{X}_{8})]=[10, 14], \mathcal{D}(\mathbb{X}_{8}) = 10,[\underline{t}_{\mathfrak{b}_4},\bar{t}_{\mathfrak{{b}}_4}] = [10, 24]$;
 
	\item $[\underline{t}_s(\mathbb{X}_{9}), \bar{t}_s(\mathbb{X}_{9})]=[15, 24], \mathcal{D}(\mathbb{X}_{9}) = 0,[\underline{t}_{\mathfrak{b}_5},\bar{t}_{\mathfrak{{b}}_5}] = [15, 24]$.
\end{itemize}

Considering the velocity constraints, the initial CBFs are designed as follows:
\begin{equation}  \label{eq:cbfs_single_integrator}
\begin{aligned}
    & \mathfrak{{b}}_1(x,t) =  (11.5 - t)^2 - \max\{|x_1 + 3|, |x_2 + 3|\}^2, t\in [0,16]; \\
    & \mathfrak{{b}}_2(x,t) =  \begin{cases}
    (11.5 - t)^2 - \max\{|x_1 + 3|, |x_2 + 3|\}^2,\\ t\in [2,18]; \\
    3^2 - \max\{|x_1 + 3|, |x_2 + 3|\}^2,  t\in [18,26];
    \end{cases} \\
     & \mathfrak{{b}}_3(x,t) =  
    (15.5 - t)^2 - \max\{|x_1 - 3|, |x_2 + 3|\}^2, t\in [10,14]; \\
    & \mathfrak{{b}}_4(x,t) =  \begin{cases}
    (15.5 - t)^2 - \max\{|x_1 - 3|, |x_2 + 3|\}^2, \\ t\in [10,14]; \\
      3^2 - \max\{|x_1 - 3|, |x_2 + 3|\}^2, t\in [14,20];
    \end{cases} \\
    & \mathfrak{{b}}_5(x,t) =  (27 - t)^2 - \max\{|x_1 - 3|, |x_2|\}^2, t\in [15,24].
\end{aligned}
\end{equation}

It is evident that the zero super-level sets of the barriers are square, which either remain static or move at a velocity of $1$. If the robot is about to leave the safe region, i.e., when $\mathfrak{{b}}_i(x,t) = 0$, the robot can always steer itself towards the center with unit velocity, and thus always stay safe.
One could easily verify that, for $i= 1,2,...,5$, 1) $\mathfrak{b}_i(x, t)$ is a valid CBF for the single integrator dynamics in (21); 2)  $\mathfrak{b}_i(x, t) = h_{\mathbb{X}_{f_i}}(x), \forall t\in [\bar{t}_s(\mathbb{X}_{f_i}), {t}_e(\mathbb{X}_{f_i})]$, where $\mathbb{X}_{f_i}$ is the set node in the corresponding temporal fragment $f_i$; 3) $  \mathfrak{b}_i(x,\underline{t}_{\mathfrak{b}_i}) \ge \mathfrak{b}_j(x,\underline{t}_{\mathfrak{b}_i}), \forall x$, where the corresponding temporal fragment $f_j$ is the predecessor of $f_i$. Thus, CBFs in \eqref{eq:cbfs_single_integrator} fulfill the conditions in Sec. III.D. Note that here we calculate the initial CBFs, which will be updated online according to algorithm 4 in \cite{8}. 

Now we demonstrate the numerical results with the proposed CBF-based QP control synthesis scheme in \eqref{QP gen}. In Fig. \ref{Fig:5}, we illustrate  trajectories with time snapshots starting from $(-2,3.5)$, both of which lie within $\mathbb{X}_1 \cap \mathbb{X}_2$. For all the trajectories, the input bound $u_k$ is respected. We observe that every trajectory satisfies the STL specification $\varphi$.
\noindent
\begin{figure}
\centering	\includegraphics[width=0.5\textwidth]{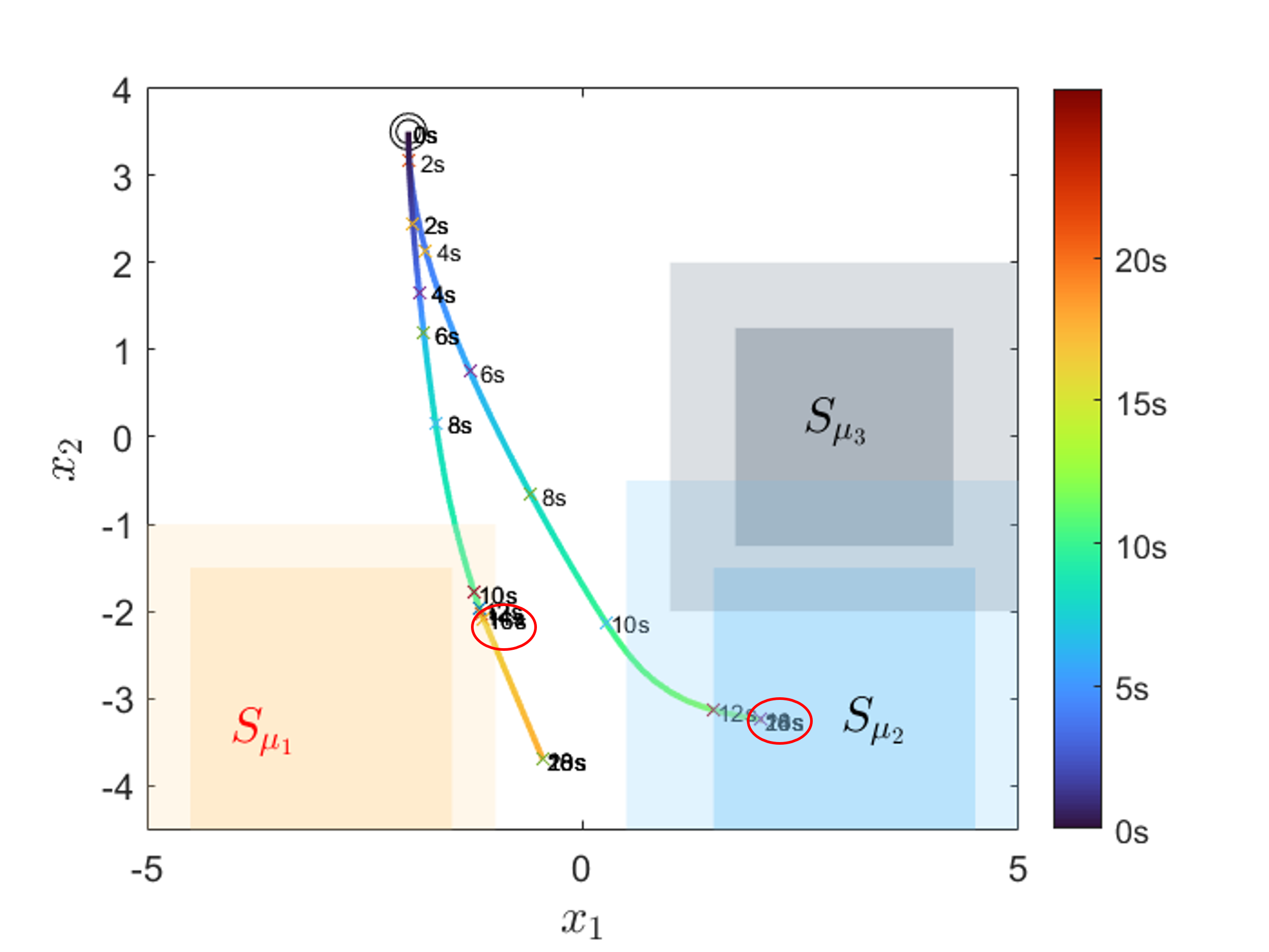}
	\caption{\footnotesize Two trajectories for the unicycle dynamics system. The yellow, blue, and gray regions denote the target regions \( S_{\mu_1} \), \( S_{\mu_2} \), and \( S_{\mu_3} \), respectively, while the light yellow, light blue, and light gray regions represent their corresponding fault-tolerant feasible sets \( X'_{t_k} \). The red circle marks a state \( x_{t_k} \notin X'_{t_k} \), indicating that the monitor can predict task failure in advance. }
 \label{Fig:4}
\end{figure}
\begin{figure}
\centering	\includegraphics[width=0.5\textwidth]{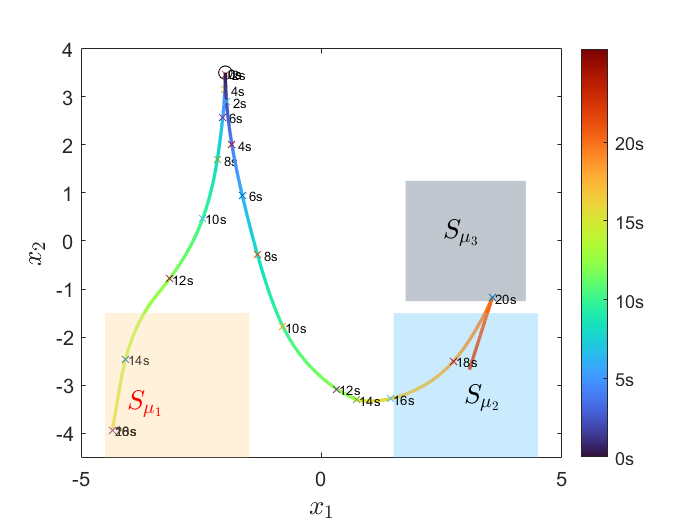}
	\caption{\footnotesize Two trajectories of a mobile robot with unicycle dynamics are synthesized using the proposed method under the STL specification $\varphi=\mathsf{G}_{[0, 16]}\mathsf{F}_{[2, 10]}\mu_1 \vee \mathsf{F}_{[10, 14]}( \mu_2  \mathsf{U}_{[5, 10]}\mu_3)$. The yellow, blue, and gray regions denote the target regions \( S_{\mu_1} \), \( S_{\mu_2} \), and \( S_{\mu_3} \). }
 \label{Fig:6}
\end{figure}
\subsection{Unicycle model}
Consider a mobile robot with a unicycle dynamics 
\begin{equation} \label{eq:unicycle}
    \begin{aligned}
    \dot{x}_1&=v\cos\theta, \\
\dot{x}_2&=v\sin\theta, \\
\dot{\theta}&=\omega,
    \end{aligned}
\end{equation}
and we discretize it wth sampling time 1s, where the state $x = (x_1,x_2, \theta)$,  the control input $u=(v,\omega)$. Here $(x_1,x_2)$ denotes the position, $\theta$ the heading angle, and $v$ the velocity, $\omega$ the turning rate. We assume that the control input $u=(v,\omega) \in u_k = \{ u\mid |v|\le 1, |\omega|\leq 1\}$. The STL task specification is again given by $\varphi=\mathsf{G}_{[0, 16]}\mathsf{F}_{[2, 10]}\mu_1 \vee \mathsf{F}_{[10, 14]}( \mu_2  \mathsf{U}_{[5, 10]}\mu_3)$, where $\mathbb{S}_{\mu_1}=\{x\in \mathbb{R}^2 \times S^1\mid  x_1 \in [-4.5\ {-1.5}] \wedge x_2 \in [-4.5\ {-1.5}]\}$,  $\mathbb{S}_{\mu_2}=\{x\in \mathbb{R}^2 \times S^1\mid  x_1 \in [1.5\ {4.5}] \wedge x_2 \in [-4.5\ {-1.5}]\}$, and $\mathbb{S}_{\mu_3}=\{x\in \mathbb{R}^2 \times S^1 \mid x_1 \in [1.75\ {4.25}] \wedge x_2 \in [-1.25\ {1.25}] \}$.

We note that the temporal fragments, their time encodings,  the time domains for the barrier functions, and the branch choosing guidelines are similar to those as in the case of single integrator dynamics and thus omitted here. The offline computation results are shown in Fig. 5. For three complete paths in the example, since $ x_{t_k} \notin X_{t_k}^{I'} $ (The red circle is shown in Figure 5), we can conclude immediately that the formula will be violated inevitably since there exists no controller under which the STLT formula can be satisfied.

The numerical results with the proposed scheme for unicycle dynamics are shown in Fig. \ref{Fig:6}. Here we illustrate the trajectories with time snapshots starting from $(-2,3.5,\pi/2)$, both of which lie within $\mathbb{X}_1 \cap \mathbb{X}_2$. An intuitive nominal controller similar to the single integrator case is also utilized in this example. For all the trajectories, the input bound $u_k$ is respected. Again, we observe that every trajectory satisfies the STL specification $\varphi$.

\section{CONCLUSIONS}

In this paper, we presented a novel framework for the collaborative design of Fault Diagnosis (FD) and Fault Tolerant Control (FTC) tailored for safety-critical Cyber-Physical Systems (CPS). By integrating both FD and FTC processes, the framework enhances the precision of fault detection while maintaining system fault tolerance capabilities. The approach was structured around four key steps: controller design, fault detector design, performance evaluation of fault detection, and the synthesis of a fault-tolerant control strategy.

The effectiveness of the proposed approach was validated through simulation studies, which demonstrated the framework’s ability to detect faults rapidly and with high accuracy, while simultaneously ensuring the stability and reliability of the fault-tolerant control system. This collaborative FD and FTC design framework offers significant potential for enhancing the resilience and safety of CPS, particularly in applications where performance under fault conditions is critical.


\bibliographystyle{unsrt}        
\bibliography{autosam.bib}           



\appendix
\end{document}